\documentclass[11pt,a4paper]{article}
\usepackage[T1]{fontenc}
\usepackage[utf8]{inputenc}
\usepackage{amsmath,amsthm,mathtools}
\usepackage{newtxtext,newtxmath}
\usepackage{aliascnt}
\usepackage[a4paper,left=28mm,right=28mm,top=26mm,bottom=27mm]{geometry}
\usepackage{microtype}
\usepackage{enumitem,booktabs,array,float}
\usepackage[hidelinks]{hyperref}
\usepackage[nameinlink,capitalise,noabbrev]{cleveref}
\setlist[enumerate]{label=\textup{(\roman*)},leftmargin=2.2em,itemsep=3pt,topsep=4pt}
\allowdisplaybreaks[2]
\numberwithin{equation}{section}
\newtheorem{theorem}{Theorem}[section]
\newaliascnt{proposition}{theorem}
\newtheorem{proposition}[proposition]{Proposition}
\aliascntresetthe{proposition}
\newaliascnt{lemma}{theorem}
\newtheorem{lemma}[lemma]{Lemma}
\aliascntresetthe{lemma}
\newaliascnt{corollary}{theorem}
\newtheorem{corollary}[corollary]{Corollary}
\aliascntresetthe{corollary}
\theoremstyle{definition}
\newaliascnt{definition}{theorem}
\newtheorem{definition}[definition]{Definition}
\aliascntresetthe{definition}
\newaliascnt{example}{theorem}
\newtheorem{example}[example]{Example}
\aliascntresetthe{example}
\theoremstyle{remark}
\newaliascnt{remark}{theorem}
\newtheorem{remark}[remark]{Remark}
\aliascntresetthe{remark}
\crefname{theorem}{Theorem}{Theorems}
\crefname{proposition}{Proposition}{Propositions}
\crefname{lemma}{Lemma}{Lemmas}
\crefname{corollary}{Corollary}{Corollarys}
\crefname{definition}{Definition}{Definitions}
\crefname{example}{Example}{Examples}
\crefname{remark}{Remark}{Remarks}
\DeclareMathOperator{\Supp}{Supp}
\DeclareMathOperator{\wt}{wt_H}
\DeclareMathOperator{\spn}{span}
\DeclareMathOperator{\sper}{sp}
\DeclareMathOperator{\mult}{mult}
\DeclareMathOperator{\Aff}{Aff}
\DeclareMathOperator{\AGL}{AGL}
\DeclareMathOperator{\GL}{GL}
\DeclareMathOperator{\Spec}{Spec}
\DeclareMathOperator{\Lie}{Lie}
\DeclareMathOperator{\rank}{rank}
\DeclareMathOperator{\ord}{ord}
\DeclareMathOperator{\Tr}{Tr}
\DeclareMathOperator{\row}{row}
\DeclareMathOperator{\Hull}{Hull}
\DeclareMathOperator{\Fix}{Fix}

\newcommand{\F}{\mathbb F}
\newcommand{\Z}{\mathbb Z}
\newcommand{\Gm}{\mathbb G_m}
\newcommand{\Ga}{\mathbb G_a}
\newcommand{\Af}{\mathcal A_f}
\newcommand{\Ah}{\mathcal A_h}
\newcommand{\Bf}{\mathcal B_f}
\newcommand{\If}{I_f}
\newcommand{\Ih}{I_h}
\newcommand{\fsp}{f_{\mathrm{sp}}}
\newcommand{\Csp}{C_{\mathrm{sp}}}
\newcommand{\bsp}{b_{\mathrm{sp}}}
\newcommand{\bsch}{B_{\mathrm{sch}}}
\newcommand{\Es}{E_{\mathrm{sch}}}
\newcommand{\Wsch}{W_{\mathcal A,\tau}}

\newcommand{\ceil}[1]{\left\lceil #1\right\rceil}
\newcommand{\floor}[1]{\left\lfloor #1\right\rfloor}
\newcommand{\arxivref}[1]{\href{https://arxiv.org/abs/#1}{arXiv:#1}}
\newcommand{\doiref}[1]{\href{https://doi.org/#1}{doi: #1}}
\hypersetup{pdftitle={Affine Symmetries of Polynomial Roots and Applications to Constacyclic Codes},pdfsubject={Affine stabilizer groups, affine stabilizer group schemes, and constacyclic codes},pdfauthor={}}
\begin{document}

\title{Affine Symmetries of Polynomial Roots with Multiplicities and Applications to Constacyclic Codes}

\author{Li Zhu$^{1}$ and Hongfeng Wu$^{2}$\footnote{Corresponding author.}
	\setcounter{footnote}{-1}
	\footnote{E-Mail addresses:
		lizhumath@pku.edu.cn (L. Zhu), whfmath@gmail.com (H. Wu)}
	\\[0.5ex]
	\small $^{1}$School of Mathematical Sciences, Guizhou Normal University, Guiyang, China
	\\
	\small $^{2}$College of Science, North China University of Technology, Beijing, China}

\date{}

\maketitle

\begin{abstract}
Let $f(X)$ be a polynomial over a finite field, and let its defining
function record the multiplicities of its roots. In this paper, we
introduce the affine stabilizer group and the affine stabilizer group
scheme of the defining function of $f(X)$. These objects characterize
the affine symmetries of the roots of $f(X)$ together with their
multiplicities. We establish structural descriptions of these affine stabilizers and
relate them to the coefficient structure of $f(X)$. In particular, the
affine stabilizer group determines quotient structures of the defining
function, which yield canonical descriptions of the corresponding
semi-invariant polynomial spaces. The resulting coefficient constraints
provide explicit restrictions on the possible supports of polynomials
with prescribed affine symmetries. The affine stabilizer group scheme
further captures infinitesimal symmetries arising from multiple roots and
gives refined coefficient constraints beyond those obtained from ordinary
affine transformations. We apply this framework to generator polynomials of constacyclic codes
over finite fields. The affine stabilizer structure of a generator
polynomial induces subcodes defined by semi-invariant conditions and
provides constraints on their dimensions and Hamming distances. In the
special case of multiplicative symmetries, the construction specializes
to the support-primitive decomposition of constacyclic codes.
\end{abstract}

\section{Introduction}\label{sec:intro}
Let $k=\F_q$ be a finite field of characteristic $p$, and let $f\in k[X]$ be a monic nonconstant polynomial with $f(0)\ne0$. An affine symmetry of the roots of $f$ is a transformation $x\mapsto ax+b$ that preserves each root multiplicity. If $\tau=\deg f$, this condition is equivalent to
\[
 f(aX+b)=a^\tau f(X).
\]
We study two stabilizers defined by this identity. Allowing $a,b$ to lie in an algebraic closure $\bar k$ gives the geometric affine stabilizer group $\Gamma$. Requiring the same identity over arbitrary commutative $k$-algebras defines the affine stabilizer group scheme $\Af$. We describe their structures and the corresponding restrictions on polynomial coefficients, and apply the results to constacyclic codes.

For the group $\Gamma$, suppose first that $f$ has at least two distinct roots. Let $V$ be its translation subgroup, and let $A$ be the image of the homomorphism $(a,b)\mapsto a$. Then $V$ is a finite $\F_p$-space and $A$ is cyclic of order $m$ prime to $p$. Writing $Q=|V|$, there is a center $c$, unique modulo $V$ when $m>1$, such that
\[
 \Gamma=\{(a,(1-a)c+v):a\in A,\ v\in V\}.
\]
The subspace polynomial $L_V(X)=\prod_{v\in V}(X-v)$ and its translate $Z(X)=L_V(X-c)$ belong to $k[X]$. Write $f=h^e$, where $e$ is the largest common $p$-power divisor of the root multiplicities and $h'\ne0$. The orbit decomposition yields
\begin{equation}\label{eq:intro-normal}
 f(X)=\bigl(Z(X)^\kappa H_0(Z(X)^m)\bigr)^e,
 \qquad H_0\in k[Y],\quad H_0(0)\ne0.
\end{equation}
Here $\kappa$ is the multiplicity of each point of $c+V$ in $h$, with value zero when this coset contains no root. The other root orbits and their multiplicities determine $H_0$. We construct this polynomial from the defining function of $f$ and describe the irreducible factors of \eqref{eq:intro-normal} through the Frobenius action on the roots.

When $f$ is irreducible of degree $d\ge2$, its roots form one Frobenius orbit. We prove that
\[
 |V|\in\{1,p\}\quad\text{if }p\text{ is odd},\qquad
 |V|\in\{1,2,4\}\quad\text{if }p=2.
\]
We determine $V$ by Frobenius-difference conditions and recover the slope group from multiplicative orders after passing to the translation quotient. This gives classifications and counts for irreducible polynomials of prime degree with nontrivial geometric affine stabilizer. We also describe the stabilizers of irreducible quartics in characteristic two.

For the group scheme $\Af$, let $K$ be a finite splitting field and write $d=\deg h$. Its local structure depends on the number of roots of $h$ whose multiplicities are not divisible by $p$. If there are at least two such roots, then $\Ah$ is finite \'etale, and the local coordinate algebra of $\Af$ at every geometric point is
\[
 K[U,W]/(U^e,W^e).
\]
If there is exactly one such root $c_*$ and $f$ has at least two distinct roots, then, in coordinates centered at $c_*$,
\[
 \Af\simeq\alpha_e\rtimes\mu_{eD_h},
 \qquad
 D_h=\gcd\{d-j:0\le j<d,\ [X^j]h(X+c_*)\ne0\}.
\]
Here $\alpha_e$ is the kernel of $b\mapsto b^e$ on the additive group scheme, and $\mu_n$ is the group scheme of $n$th roots of unity. We also treat polynomials with one distinct root. From these descriptions we obtain semi-invariant spaces for $\Gamma$ and $\Af$ and compute their support envelopes in the original variable $X$.

The multiplicative part is expressed in terms of the defining function. Fix a primitive $M$th root of unity $\zeta_M$, where $p\nmid M$, and write
\[
 f(X)=\prod_{\gamma\in T}(X-\zeta_M^\gamma)^{\Omega(\gamma)},
 \qquad
 \Sigma_\Omega=\{t\in\Z/M\Z:\Omega(\gamma+t)=\Omega(\gamma)\text{ for all }\gamma\}.
\]
The function $\Omega$ is extended by zero outside $T$. Let $s=\sper(f)$ denote the greatest common divisor of the positive indices of the nonzero coefficients of $f$. Then
\[
 s=e|\Sigma_\Omega|,\qquad f(X)=\fsp(X^s),\qquad\sper(\fsp)=1.
\]
For a constacyclic code $C=(f)\subseteq k[X]/(X^N-\lambda)$, this gives an interleaved direct sum of $s$ copies of a shorter code. We also construct subcodes from affine semi-invariant powers. Their dimensions are computed by an outer divisibility condition, and their supports are determined by a row-space criterion. Applying the Singleton bound to these subcodes gives upper bounds on $d_H(C)$.

Translation-invariant polynomials and linear fractional actions on irreducible polynomials over finite fields were studied by Reis \cite{ReisGL}; compositions with linearized polynomials were considered in \cite{ReisComp}. Schulz \cite{Schulz} relates invariant polynomials and their factorizations to quotient maps for finite projective linear groups. The exponent representation used below employs Zech logarithms; see \cite{Zech}. For cyclic codes, Kashiwagi and Kikumasa \cite[Theorem~8]{KK} constructed the indecomposable coordinate components from the support-index greatest common divisor, without a squarefreeness assumption. G\'omez Macedo and Zald\'ivar \cite{GMZ} established uniqueness and structural results for group-code decompositions and described the interleaving of cyclic components. The multiple equal-difference descriptions of defining sets and the arithmetic Singleton bound are given in \cite{MED,AS}. In the present paper, the root multiplicities determine the affine stabilizers and the corresponding coefficient spaces; their origin-preserving scaling restriction gives the root-side description of those coordinate components.

The paper is organized as follows. Section~\ref{sec:prelim} fixes the notation and recalls the preliminary results. Section~\ref{sec:quotients} describes the affine stabilizer group, its root orbits and its semi-invariant spaces; the irreducible case is treated in Subsection~\ref{sec:irreducible}. Section~\ref{sec:schemes} determines the affine stabilizer group scheme and its local structure, and derives its semi-invariant spaces. Section~\ref{sec:codes} applies these results to constacyclic codes. Section~\ref{sec:examples} contains examples. Appendix~\ref{app:pairing} gives a coefficient-pairing and recurrence method for the outer ranks.
\section{Preliminaries}\label{sec:prelim}
\subsection{Defining functions and root multiplicities}
Throughout, $k=\F_q$ with $q=p^a$, and $\bar k$ is a fixed algebraic closure. For every positive integer $M$ prime to $p$, fix compatible primitive roots of unity $\zeta_M$: if $M\mid M'$, then $\zeta_{M'}^{M'/M}=\zeta_M$. A defining function is a nonzero function
\begin{equation}\label{eq:defining}
 \Omega:\Z/M\Z\longrightarrow\Z_{\ge0},\qquad
 \Omega(q\gamma)=\Omega(\gamma).
\end{equation}
Its positive support is $T=\{\gamma:\Omega(\gamma)>0\}$, and it determines
\begin{equation}\label{eq:f-root}
 f(X)=\prod_{\gamma\in T}(X-\zeta_M^\gamma)^{\Omega(\gamma)},
 \qquad \tau=\deg f=\sum_{\gamma\in T}\Omega(\gamma).
\end{equation}
The Frobenius condition in \eqref{eq:defining} is equivalent to $f\in k[X]$. Every monic nonconstant polynomial over $k$ with nonzero constant term has such a description, since its nonzero roots have finite multiplicative order prime to $p$.

Fix a finite splitting field $K\supseteq k(\zeta_M)$ and write
\[
 S=\{\zeta_M^\gamma:\gamma\in T\},\qquad
 \Omega_K(\alpha)=\mult_\alpha(f)
 =\max\{j:(X-\alpha)^j\mid f\}.
\]
We extend $\Omega_K$ by zero to $K$. Throughout the paper, an affine symmetry of polynomial roots is required to preserve their integer multiplicities. In particular, $\Omega_K(0)=0$, whereas $\Omega(0)$ denotes the multiplicity of the root $1$.

For a polynomial $g=\sum_jg_jX^j$, define
\[
 \Supp(g)=\{j:g_j\ne0\},\qquad \wt(g)=|\Supp(g)|.
\]
All supports and weights are measured in the original coordinate $X$, unless another coordinate is explicitly specified. For $g$ nonconstant with $g(0)\ne0$, put
\[
 \sper(g)=\gcd\{j>0:g_j\ne0\}.
\]
 For $t\ge0$, $k[X]_{<t}$ and $k[X]_{\le t}$ denote the indicated degree-filtered vector spaces.

\begin{lemma}\label{lem:power}
Set
\[
 v=\min_{\gamma\in T}v_p(\Omega(\gamma)),\qquad e=p^v.
\]
There is a unique monic $h\in k[X]$ with
\[
 f=h^e,\qquad h'\ne0,\qquad
 \mult_{\zeta_M^\gamma}(h)=\Omega(\gamma)/e.
\]
If $d=\deg h=\tau/e$, then $\Supp(f)=e\Supp(h)$ and $\wt(f)=\wt(h)$.
\end{lemma}
\begin{proof}
The root product with multiplicities $\Omega/e$ is Frobenius-invariant, hence lies in $k[X]$, and its $e$th power is $f$. Frobenius is injective on a field, which proves uniqueness. If $h'=0$, all positive support indices of $h$ are divisible by $p$. Since $k$ is perfect, $h$ would be a $p$th power, contrary to the definition of $e$. Finally, $h=\sum_jh_jX^j$ implies $f=\sum_jh_j^eX^{ej}$.
\end{proof}
The polynomial $h$ may have repeated roots, with multiplicities $\Omega(\gamma)/e$. For $H(Y)=\sum_iH_iY^i$, we write
\[
 H^{[e]}(Y)=\sum_iH_i^eY^i.
\]
This notation denotes coefficientwise Frobenius, not the polynomial power $H(Y)^e$.

The $j$th Hasse derivative is
\[
 \partial^{[j]}g(X)=\sum_{i\ge j}\binom ijg_iX^{i-j}.
\]
It is characterized by $g(X+B)=\sum_j\partial^{[j]}g(X)B^j$, and
\begin{equation}\label{eq:hasse}
 \mult_\alpha(g)\ge r
 \quad\Longleftrightarrow\quad
 \partial^{[j]}g(\alpha)=0\quad(0\le j<r).
\end{equation}

\subsection{Subspace polynomials and finite-field descent}
The following properties of subspace polynomials will be used below; see also \cite{ReisGL,GowMcG}.
\begin{lemma}\label{lem:subspace}
Let $V$ be a finite $\F_p$-subspace of a field of characteristic $p$, and let $Q=|V|=p^r$. Then
\[
 L_V(X)=\prod_{v\in V}(X-v)
       =X^{p^r}+\sum_{j=0}^{r-1}a_jX^{p^j},
 \qquad a_0\ne0
\]
when $r>0$, and $L_{\{0\}}(X)=X$. It is additive. If $uV=V$, then
\begin{equation}\label{eq:scalar-L}
 L_V(uX)=uL_V(X).
\end{equation}
If $V\ne0$ and a finite group $A\le\bar k^\times$ of order $m$ preserves $V$, then $m\mid Q-1$. Over any field $F$ containing $V$,
\begin{equation}\label{eq:translation-ring}
 F[X]^V=F[L_V(X)].
\end{equation}
If $V$ is stable under $q$th powers, then $L_V\in k[X]$.
\end{lemma}
\begin{proof}
If $V=V'\oplus\F_pv$ and $b=L_{V'}(v)\ne0$, induction gives
\[
 L_V(X)=\prod_{t\in\F_p}(L_{V'}(X)-tb)
       =L_{V'}(X)^p-b^{p-1}L_{V'}(X).
\]
This proves additivity and the nonzero linear coefficient. The root product gives $L_V(uX)=u^QL_V(X)$; comparing linear coefficients gives $u^Q=u$, proving \eqref{eq:scalar-L} and the assertion about $m$. To prove \eqref{eq:translation-ring}, divide a translation-invariant polynomial by the monic polynomial $L_V$. Uniqueness of quotient and remainder makes both invariant. An invariant remainder of degree less than $Q$ is constant, since its difference from its value at zero vanishes at every point of $V$. Induction finishes the proof. Frobenius descent follows by permuting the roots of $L_V$.
\end{proof}
In particular, $L_{b\F_p}(X)=X^p-b^{p-1}X$ for $b\ne0$. An arbitrary subspace need not be a scalar multiple of a subfield, so its subspace polynomial need not be a binomial.

We write $F(\alpha)=\alpha^q$ for Frobenius. A Frobenius orbit of length $r$ is the root set of one monic irreducible polynomial of degree $r$ over $k$. We use this correspondence also on quotient root sets. These standard finite-field facts can be found in \cite{MilneFT}.

\subsection{Affine transformations and group schemes}
We represent an affine map by $(a,b)$, acting as $X\mapsto aX+b$, with
\[
 (a,b)(a',b')=(aa',ab'+b).
\]
Thus $\Aff_1$ is the semidirect product $\Ga\rtimes\Gm$. A polynomial $g$ satisfying $g(aX+b)=a^tg(X)$ is a semi-invariant of the degree character $\chi_t(a,b)=a^t$; the character need not be trivial. Substitution is a right action on polynomials, which does not affect these equations.

For a commutative $k$-algebra $R$, the group schemes needed below have points
\[
 \Ga(R)=(R,+),\quad \Gm(R)=R^\times,\quad
 \mu_n(R)=\{a:a^n=1\},\quad
 \alpha_e(R)=\{b:b^e=0\}.
\]
Here $e$ is a power of $p$. For a finite group scheme, its rank is the dimension of its coordinate algebra over $k$. Basic terminology and the distinction between finite \emph{\'etale} and nonreduced group schemes are as in \cite{MilneAGS}.

\subsection{Constacyclic codes}
A $\lambda$-constacyclic code is an ideal of
\[
 R_{N,\lambda}=k[X]/(X^N-\lambda),\qquad\lambda\in k^\times.
\]
Its unique monic generator $f$ divides $X^N-\lambda$. We restrict to $0<\deg f<N$, and identify codewords with their unique representatives of degree less than $N$. A coordinate decomposition means a direct sum of nonzero subcodes with pairwise disjoint supports. Kashiwagi and Kikumasa \cite[Theorem~8]{KK} show that a nontrivial cyclic code is indecomposable precisely when its generator has support-index greatest common divisor one. The corresponding constacyclic statement is given in \cref{thm:decomposition}. Standard metric facts for linear codes are used as in \cite{HP}.

\section{The affine stabilizer group of a polynomial}\label{sec:quotients}
This section describes the affine stabilizer group, its root orbits and the associated polynomial expressions. We then consider irreducible polynomials and the coefficient spaces defined by the group action.

\subsection{Definition and structure}
\begin{definition}\label{def:affine-group}
The affine stabilizer group of $f$ is
\begin{equation}\label{eq:Gamma}
 \Gamma_f=\{(a,b)\in\bar k^\times\times\bar k:
       \mult_{a\alpha+b}(f)=\mult_\alpha(f)\text{ for every }\alpha\in S\}.
\end{equation}
\end{definition}
We write $\Gamma=\Gamma_f$. The group law is composition of affine transformations. Unless stated otherwise, this section assumes that $f$ has at least two distinct roots.

Assume $|S|\ge2$. Define
\[
 V=\{b:(1,b)\in\Gamma\},\qquad
 A=\{a:\text{there is }b\text{ with }(a,b)\in\Gamma\},\qquad
 M_0=\{a:(a,0)\in\Gamma\},
\]
and set $Q=|V|$, $m=|A|$, and $\sigma=|M_0|$.

\begin{proposition}\label{prop:group}
The group $\Gamma$ is finite, and all its parameters lie in $K$. There is an exact sequence
\[
 1\longrightarrow V\longrightarrow\Gamma\longrightarrow A\longrightarrow1,
\]
where $V$ is an $\F_p$-space, $A$ is cyclic of order prime to $p$, and $aV=V$ for $a\in A$. Both $Q\mid\tau$ and $Q\mid d$. These groups are stable under Frobenius, and they are unchanged when $f$ is replaced by $h$.
\end{proposition}
\begin{proof}
An affine transformation is determined by the images of two distinct roots, so its parameters are differences and ratios of elements of $K$. Its permutation action on $S$ is faithful, hence finite. The kernel of the slope homomorphism is an additive subgroup, and
\[
 (a,b)(1,v)(a,b)^{-1}=(1,av).
\]
The slope image is a subgroup of the multiplicative group of a finite field. Translations act freely on the roots, and the multiplicity is constant on every $V$-coset; summing these multiplicities gives the divisibilities. Frobenius preserves all the defining conditions. Dividing every multiplicity by the same positive integer does not change any stabilizer.
\end{proof}

\begin{theorem}\label{thm:center}
If $m>1$, there is $c\in K$, unique modulo $V$, such that
\begin{equation}\label{eq:center-group}
 \Gamma=\{(a,(1-a)c+v):a\in A,\ v\in V\}.
\end{equation}
If $m=1$, choose $c=0$. In both cases
\begin{equation}\label{eq:center-poly}
 L_V\in k[X],\qquad Z(X):=L_V(X-c)\in k[X],\qquad
 Z(aX+(1-a)c+v)=aZ(X).
\end{equation}
For $m>1$, $M_0=A$ if $c\in V$, and $M_0=\{1\}$ otherwise.
\end{theorem}
\begin{proof}
Choose a generator $a\ne1$ of $A$ and a lift $(a,b)$. Its fixed point is $c=b/(1-a)$, and its $m$th power is the identity. Removing suitable powers of this lift from any group element leaves a translation, proving \eqref{eq:center-group}. Two lifts of the same slope differ by an element of $V$. The operator $1-a$ is bijective on $V$: it is injective and maps $V$ to itself. Hence the corresponding centers differ by an element of $V$.

Frobenius preserves $\Gamma$, so uniqueness gives $c^q-c\in V$. Also $L_V\in k[X]$ by \cref{lem:subspace}. Consequently $L_V(c)^q=L_V(c^q)=L_V(c)$, and
\[
 Z(X)=L_V(X)-L_V(c)\in k[X].
\]
Its transformation law follows from additivity and \eqref{eq:scalar-L}. Finally, $(a,0)\in\Gamma$ is equivalent to $(1-a)c\in V$. For $a\ne1$, bijectivity of $1-a$ on $V$ makes this equivalent to $c\in V$.
\end{proof}
The center representative need not lie in $k$. The descended object is the coset $c+V$, or equivalently the polynomial $Z$.

\subsection{Defining functions and affine transformations}

Put $n_K=|K^\times|$ and choose $\xi=\zeta_{n_K}$, so that $\zeta_M=\xi^{n_K/M}$. The augmented index set and its realization are
\[
 I_K=(\Z/n_K\Z)\sqcup\{\bot\},\qquad
 \epsilon(t)=\xi^t\ (t\ne\bot),\quad \epsilon(\bot)=0.
\]
In particular, the index $0$ represents $1$, not the zero of the field. Embed $\Z/M\Z$ by $\iota(\gamma)=(n_K/M)\gamma$ and define $\widetilde\Omega=\Omega_K\circ\epsilon$ on $I_K$.

Define the augmented Zech permutation $\mathcal Z_\xi$ by
\[
 \epsilon(\mathcal Z_\xi(t))=1+\epsilon(t).
\]
This includes $\mathcal Z_\xi(\bot)=0$ and the index of $-1$ mapping to $\bot$. It satisfies $\mathcal Z_\xi^p=1$. For $u\in\Z/n_K\Z$, let $s_u(t)=t+u$ on ordinary indices and $s_u(\bot)=\bot$. For $v\in I_K$, set
\begin{equation}\label{eq:index-affine}
 \mathcal T_{u,v}=
 \begin{cases}
 s_u,&v=\bot,\\
 s_v\mathcal Z_\xi s_{u-v},&v\ne\bot.
 \end{cases}
\end{equation}
Then
\[
 \epsilon(\mathcal T_{u,v}(t))=\xi^u\epsilon(t)+\epsilon(v).
\]
Thus the permutations $\mathcal T_{u,v}$ form a group conjugate to $\Aff_1(K)$ under $\epsilon$.

\begin{proposition}\label{prop:index}
Let $a=\xi^u$ and $b=\epsilon(v)$. The following conditions are equivalent:
\[
 \widetilde\Omega\circ\mathcal T_{u,v}=\widetilde\Omega,
 \qquad \Omega_K(ax+b)=\Omega_K(x)\quad(x\in K),
 \qquad f(aX+b)=a^\tau f(X).
\]
The map $\mathcal T_{u,v}\mapsto(\xi^u,\epsilon(v))$ identifies the stabilizer of $\widetilde\Omega$ with $\Gamma$.
The original index stabilizer
\[
 \Sigma_\Omega=\{t\in\Z/M\Z:\Omega(\gamma+t)=\Omega(\gamma)
                                      \text{ for all }\gamma\}
\]
is exactly the subgroup of origin-preserving scalings:
\begin{equation}\label{eq:M0}
 M_0=\{a:(a,0)\in\Gamma\}
      =\{\zeta_M^t:t\in\Sigma_\Omega\}.
\end{equation}
\end{proposition}
\begin{proof}
The first two conditions are conjugate under $\epsilon$. An invertible affine substitution preserves root multiplicities, and two polynomials with the same roots and multiplicities differ by a nonzero scalar. Comparing leading coefficients gives the scalar $a^\tau$. For the last assertion, any scaling preserving the nonempty root set is the ratio of two roots, hence belongs to $\mu_M(K)$. Its action on the defining function is precisely index translation.
\end{proof}
Only the positive support needs to be checked for any candidate permutation: injectivity and finiteness then force preservation of the zero complement.

Frobenius acts on indices by $t\mapsto qt$, fixing $\bot$, and commutes with $\mathcal Z_\xi$. It conjugates $\mathcal T_{u,v}$ to $\mathcal T_{qu,qv}$, with the evident convention at $\bot$. A change $\xi'=\xi^c$, $\gcd(c,n_K)=1$, simultaneously conjugates the Zech permutation and the stabilizer by $t\mapsto ct$. These simultaneous changes conjugate the defining-function stabilizers. In contrast, replacing root-index translation by $\gamma\mapsto a\gamma+b$ would describe a power map on roots, not general root-value affine addition.

The group $\Gamma$ can be determined from the defining function as follows. Fix two different roots indexed by $\gamma_0,\gamma_1$. Any pair of distinct candidate images indexed by $\eta_0,\eta_1$, with the matching multiplicities, determines
\[
 a=\frac{\zeta_M^{\eta_1}-\zeta_M^{\eta_0}}
          {\zeta_M^{\gamma_1}-\zeta_M^{\gamma_0}},
 \qquad b=\zeta_M^{\eta_0}-a\zeta_M^{\gamma_0}.
\]
Checking the full positive support tests whether this pair belongs to $\Gamma$. All differences and sums can be performed using the fixed exponent model. If $r=|S|$, there are at most $r(r-1)$ pairs and $r$ tests per pair. This procedure uses $O(r^3)$ field or table operations after the field data have been prepared. A translation-only search uses $O(r^2)$ such operations. A full Zech table has $|K|$ entries; the stated operation counts exclude its construction and storage.

\subsection{Root orbits and polynomial expressions}

\begin{lemma}\label{lem:orbits}
The coset $c+V$ is a $\Gamma$-orbit of size $Q$ with root polynomial $Z$. Every point $\alpha\notin c+V$ has a free orbit of size $Qm$, whose monic root polynomial is
\begin{equation}\label{eq:orbit-poly}
 \prod_{\beta\in\Gamma\alpha}(X-\beta)=Z(X)^m-Z(\alpha)^m.
\end{equation}
When $m=1$, all orbits are translation cosets; $V$ is simply the chosen distinguished orbit and is not contained in $S$.
\end{lemma}
\begin{proof}
In coordinates $Y=X-c$, the group acts as $Y\mapsto aY+v$. If a nonidentity element fixes $Y\notin V$, its slope is not one and $Y=-v/(a-1)\in V$, a contradiction. Taking the product first over a translation coset and then over $A=\mu_m(K)$ gives
\[
 \prod_{a\in A}(Z(X)-aZ(\alpha))=Z(X)^m-Z(\alpha)^m.
\]
For $m=1$, $c=0$, and the absence of a zero root excludes the orbit $V$ from $S$.
\end{proof}

\begin{theorem}\label{thm:normal}
Let $\kappa=\Omega_K(c)/e$, with value zero when the distinguished orbit is not a root orbit. There is a unique monic $H_0\in k[Y]$, with $H_0(0)\ne0$, such that
\begin{equation}\label{eq:normal}
 h(X)=Z(X)^\kappa H_0(Z(X)^m),\qquad
 f(X)=\bigl(Z(X)^\kappa H_0(Z(X)^m)\bigr)^e.
\end{equation}
Writing $r_0=\deg H_0$, one has $d=Q(\kappa+mr_0)$. The multiplicity of $H_0$ at $Z(\alpha)^m\ne0$ is $\Omega_K(\alpha)/e$.
\end{theorem}
\begin{proof}
The residual multiplicity $\Omega_K/e$ is constant on every orbit. Apply \cref{lem:orbits} to the product defining $h$. Two nondistinguished roots give the same value $Z(\alpha)^m$ precisely when their $Z$-values differ by an element of $A$, and hence when the original roots lie in the same $\Gamma$-orbit. This proves the multiplicity assertion and uniqueness. Frobenius permutes the quotient points and preserves their multiplicities. Since $Z\in k[X]$, it follows that $H_0\in k[Y]$. Degree comparison gives the remaining identity.
\end{proof}
The case $\kappa>0$ occurs when the center coset $c+V$ is contained in $S$. The factor $H_0$ may have repeated roots and may satisfy $H_0'=0$.

The quotient can be constructed directly from the defining function. For
\[
 T_{\mathrm{ns}}=\{\gamma\in T:Z(\zeta_M^\gamma)\ne0\},\qquad
 \theta(\gamma)=\log_\xi\bigl(Z(\zeta_M^\gamma)^m\bigr),
\]
define, on $\Z/n_K\Z$, the zero-extended function
\begin{equation}\label{eq:full-quotient}
 \Omega_{\mathrm{full}}(\theta(\gamma))=\Omega(\gamma)/e.
\end{equation}
The fibers are exactly the nondistinguished affine orbits, so this is well-defined and
\[
 H_0(Y)=\prod_{z\in\theta(T_{\mathrm{ns}})}(Y-\xi^z)^{\Omega_{\mathrm{full}}(z)}.
\]
Moreover, $\theta(q\gamma)=q\theta(\gamma)$. Equivalently,
\[
 \Omega_{\mathrm{full}}(z)
 =\frac1{eQm}\sum_{\substack{\gamma\in T_{\mathrm{ns}}\\\theta(\gamma)=z}}\Omega(\gamma).
\]
The full quotient data are $(Z,m,e,\kappa,\Omega_{\mathrm{full}})$, and the zero quotient point is recorded by $\kappa$. In general $\theta$ is not a homomorphism of cyclic groups.

\subsection{Translation and scaling quotients}

The subgroup
\[
 G_0=V\rtimes M_0=\{(a,v):a\in M_0,\ v\in V\}
\]
acts freely on $S$, because $S\cap V=\varnothing$. Put
\begin{equation}\label{eq:P0}
 P_0(X)=L_V(X)^{e\sigma}.
\end{equation}
The orbit products give the following proposition.

\begin{proposition}\label{prop:origin}
There is a unique monic $H_{\mathrm{or}}\in k[Y]$ such that
\begin{equation}\label{eq:origin}
 f=H_{\mathrm{or}}(P_0),\qquad
 \deg H_{\mathrm{or}}=\frac{\tau}{eQ\sigma},\qquad
 H_{\mathrm{or}}(0)\ne0,\quad H_{\mathrm{or}}'\ne0.
\end{equation}
For
\[
 \varphi_0(\gamma)=\log_\xi P_0(\zeta_M^\gamma),
 \qquad
 \Omega_{\mathrm{or}}(\varphi_0(\gamma))=\Omega(\gamma)/e,
\]
the function $\Omega_{\mathrm{or}}$ is well-defined and is the defining function of $H_{\mathrm{or}}$.
\end{proposition}
\begin{proof}
The ordinary $G_0$-orbit polynomial of a root $\alpha$ is
$L_V(X)^\sigma-L_V(\alpha)^\sigma$.
Its multiplicity is divisible by $e$. Taking the $e$th power of each orbit polynomial replaces its constant by its $e$th power, giving \eqref{eq:origin} and the quotient function. The minimum $p$-adic valuation of the residual orbit multiplicities is zero, hence $H_{\mathrm{or}}'\ne0$. Frobenius descent and uniqueness follow as in \cref{thm:normal}.
\end{proof}

\begin{theorem}\label{thm:period}
For every $f$ as in \eqref{eq:f-root}, including the single-root case,
\begin{equation}\label{eq:period}
 s:=\sper(f)=e\sigma.
\end{equation}
There is a unique support-primitive polynomial $\fsp\in k[Y]$ satisfying $f(X)=\fsp(X^s)$. If one representative $\alpha_O$ is chosen from each $M_0$-orbit in $S$, then
\begin{equation}\label{eq:core-roots}
 \fsp(Y)=\prod_{O\in S/M_0}
       (Y-\alpha_O^{e\sigma})^{\Omega_K(\alpha_O)/e}.
\end{equation}
Equivalently, on root indices the quotient is
\begin{equation}\label{eq:core-indices}
 \varphi_{\mathrm{sp}}:\Z/M\Z\longrightarrow\Z/(M/\sigma)\Z,
 \quad\gamma\longmapsto e\gamma,
 \qquad
 \Omega_{\mathrm{sp}}(\varphi_{\mathrm{sp}}(\gamma))=\Omega(\gamma)/e.
\end{equation}
\end{theorem}
\begin{proof}
For any integer $u\ge0$, all positive support indices of $f$ are divisible by $p^u$ precisely when $f$ is a $p^u$th power over the perfect field $k$. By root multiplicities this is equivalent to $p^u\mid\Omega(\gamma)$ for all $\gamma\in T$. Thus $v_p(s)=v$.

If a geometric scaling preserves the roots with multiplicities, then $f(aX)=a^\tau f(X)$. Comparing nonzero constant terms gives $a^\tau=1$, so its condition is $a^j=1$ for every support index $j$. Hence $M_0$ consists of the geometric $s$th roots of unity and has order $s/p^{v_p(s)}$. This proves \eqref{eq:period}. Dividing every support index by $s$ gives the unique $\fsp$ with support period one.

Each scaling orbit has polynomial $X^\sigma-\alpha_O^\sigma$. Its multiplicity is a multiple of $e$, so \eqref{eq:core-roots} follows by taking the $e$th Frobenius power of this factor. Distinct orbits give distinct quotient roots. The kernel of \eqref{eq:core-indices} is the unique subgroup of order $\sigma$, namely $\Sigma_\Omega$, and compatibility of primitive roots gives exactly the same quotient product.
\end{proof}
In the squarefree case $e=1$, the coarsest partition of $T$ into cosets of a common subgroup is $T/\Sigma_\Omega$. This is the multiple equal-difference (MED) description used in \cite{MED,AS}.

The relation between $P_0$ and $\fsp$ can also be expressed by polynomial composition. By \eqref{eq:scalar-L} gives a unique $R_V\in k[Y]$ with
\[
 L_V(X)=XR_V(X^\sigma),\qquad R_V(0)\ne0.
\]
If $J_f(Y)=Y(R_V^{[e]}(Y))^\sigma$, then
\begin{equation}\label{eq:core-composition}
 P_0(X)=J_f(X^s),\qquad
 \fsp(Y)=H_{\mathrm{or}}(J_f(Y)),\qquad \deg J_f=Q.
\end{equation}
This identity follows by raising $L_V=XR_V(X^\sigma)$ to the power $e\sigma$ and using uniqueness in \cref{thm:period}.

\begin{proposition}\label{prop:maximal}
Let $W\subseteq\bar k$ be a finite Frobenius-stable additive subspace, let $t$ be prime to $p$ with $\mu_t(\bar k)W=W$, and let $u\ge0$. Then
\[
 f\in k[L_W(X)^{p^ut}]
 \quad\Longleftrightarrow\quad
 W\subseteq V,\quad t\mid\sigma,\quad u\le v.
\]
Consequently $P_0$ has maximal degree among these standard quotients.
\end{proposition}
\begin{proof}
A composition of the stated form is invariant under translations by $W$ and scalings by $\mu_t$, giving the first two necessary conditions. Its outer constant term equals $f(0)\ne0$, so every root of $f$ has nonzero $L_W$-value. Since $L_W'\ne0$ and $p\nmid t$, the map $L_W^t$ is unramified at these roots. Their multiplicities in a composition with $(L_W^t)^{p^u}$ are therefore divisible by $p^u$, giving $u\le v$.

Conversely, $W\rtimes\mu_t$ acts freely on $S$. Its orbit polynomials have the form $L_W^t-a$, and their multiplicities are divisible by $p^u$. Their product lies in $K[L_W^{p^ut}]$. Since the inner polynomial belongs to $k[X]$, uniqueness of the outer polynomial gives descent to $k$.
\end{proof}

\subsection{Irreducible factorization}\label{sec:fibers}
We determine the degrees and multiplicities of the irreducible factors associated with the orbit decomposition. The following lemma relates the Frobenius action on the roots to its induced action on the quotient. For factorizations of compositions with linearized polynomials, see also \cite{ReisComp}.

\begin{lemma}\label{lem:return}
Let $F$ and $\bar F$ be permutations of finite sets $S_0$ and $T_0$, and suppose a surjection $\pi:S_0\to T_0$ satisfies $\pi F=\bar F\pi$. If $\beta\in T_0$ has $\bar F$-orbit length $r$, the cycles of
\[
 R_\beta=F^r|_{\pi^{-1}(\beta)}
\]
are in bijection with the $F$-cycles above that orbit. A return cycle of length $\ell$ corresponds to a total cycle of length $r\ell$.
\end{lemma}
\begin{proof}
For $\alpha\in\pi^{-1}(\beta)$, the equation $F^j(\alpha)=\alpha$ first implies $r\mid j$. Writing $j=ri$ reduces it to $R_\beta^i(\alpha)=\alpha$. Each total cycle intersects the chosen fiber in exactly one return cycle, and translates of this intersection recover the total cycle.
\end{proof}

\begin{theorem}\label{thm:additive-return}
Let $V$ be Frobenius-stable, $L=L_V\in k[X]$, and $Q=|V|$. Let $P\in k[Y]$ be monic irreducible of degree $r$, choose a root $\beta$, and take $\alpha$ with $L(\alpha)=\beta$. On $V$, define
\[
 A_r(v)=v^{q^r},\qquad b_\beta=\alpha^{q^r}-\alpha,\qquad
 R_\beta(v)=A_r(v)+b_\beta.
\]
Then $b_\beta\in V$. If the return cycle lengths are $\ell_1,\ldots,\ell_t$, the polynomial $P(L(X))$ has exactly $t$ distinct monic irreducible factors, of degrees $r\ell_1,\ldots,r\ell_t$, all with multiplicity one. Changing $\alpha$ within its fiber conjugates $R_\beta$ by a translation. If $G=\prod_iP_i^{a_i}$, every factor arising from $P_i(L)$ in $G(L)^e$ has multiplicity $ea_i$.
\end{theorem}
\begin{proof}
Because $L\in k[X]$, one has $L(\alpha^{q^r})=\beta$, and the fiber is $\alpha+V$. Its $q^r$-power action, in the coordinate $\alpha+v\leftrightarrow v$, is exactly $R_\beta$. Apply \cref{lem:return} and the Frobenius interpretation of irreducible factors. The derivative of $L$ is a nonzero constant, and $P$ is separable; hence $P(L)$ has no repeated roots. Distinct $P_i$ have disjoint inverse root sets. These observations prove the multiplicity statement as well.
\end{proof}
The cycle counts can be computed from fixed points. For $j\ge1$, put
\[
 b_j=\sum_{i=0}^{j-1}A_r^i(b_\beta),\qquad N_j=|\Fix(R_\beta^j)|.
\]
The equation for a fixed point is $(A_r^j-I)v=-b_j$, and therefore
\begin{equation}\label{eq:return-count}
 N_j=\begin{cases}
 p^{\dim_{\F_p}\ker(A_r^j-I)},&-b_j\in\operatorname{im}(A_r^j-I),\\
 0,&\text{otherwise}.
 \end{cases}
\end{equation}
M\"obius inversion gives the number of factors of degree $r\ell$ as
\[
 \frac1\ell\sum_{j\mid\ell}\mu(\ell/j)N_j,
\]
where $\mu$ is the integer M\"obius function. If $V\subseteq\F_{q^r}$, the return is a translation: it has $Q$ fixed points when $b_\beta=0$, and $Q/p$ cycles of length $p$ otherwise.

\begin{example}\label{ex:additive-return}
Over $\F_3$, take $L=X^3-X$. If $\beta$ is a root of an irreducible quadratic $P$, then for $L(\alpha)=\beta$,
\[
 \alpha^9-\alpha=\beta+\beta^3=\Tr_{\F_9/\F_3}(\beta).
\]
For $P_1=Y^2+1$ this trace is zero, and
\[
 P_1(L)=(X^2+1)(X^2+X+2)(X^2+2X+2).
\]
For $P_2=Y^2+Y+2$ the trace is $2$, and $P_2(L)$ is irreducible of degree six.
\end{example}

\begin{theorem}\label{thm:full-return}
Use \cref{thm:normal}, with $m>1$, and put $\pi=Z^m$. Let $P$ be a monic irreducible factor of $H_0$, let $r=\deg P$, and choose $\pi(\alpha)=\beta$ with $P(\beta)=0$. Since $\beta\ne0$, this fiber is a free transitive $\Gamma$-set. There is a unique $g_0\in\Gamma$ with $F^r(\alpha)=g_0(\alpha)$. On $\Gamma$, define
\begin{equation}\label{eq:semilinear-return}
 \vartheta_r(g)=F^rgF^{-r},\qquad
 R_\beta(g)=\vartheta_r(g)g_0.
\end{equation}
A return cycle of length $\ell$ yields one irreducible factor of $P(\pi(X))$ of degree $r\ell$. If $P$ has multiplicity $a_P$ in $H_0$, every such factor has multiplicity $ea_P$ in $f$.
\end{theorem}
\begin{proof}
The fiber is a free orbit by \cref{lem:orbits}. Frobenius normalizes $\Gamma$, and under the bijection $g\mapsto g(\alpha)$,
\[
 F^r(g(\alpha))=\vartheta_r(g)g_0(\alpha).
\]
Thus \eqref{eq:semilinear-return} is precisely the return permutation. Apply \cref{lem:return}. Since $Z'$ is a nonzero constant and $p\nmid m$, the derivative of $Z^m$ is nonzero on this nonzero fiber. Its factors are simple before the powers from $a_P$ and $e$ are imposed.
\end{proof}
The parameters in \eqref{eq:semilinear-return} are explicit: $\vartheta_r(a,b)=(a^{q^r},b^{q^r})$. The distinguished factor $Z^{e\kappa}$ must be treated separately. Write $Z=L_V-\delta$, where $\delta=L_V(c)\in k$; \cref{thm:additive-return} with $P(Y)=Y-\delta$ gives its factor degrees, and every factor has multiplicity $e\kappa$. For $m=1$, all fibers are handled additively.

\subsection{Irreducible polynomials}\label{sec:irreducible}
Throughout this subsection $f$ is monic irreducible of degree $d\ge2$, $K=\F_{q^d}$, and $\alpha$ is a root. Write $\alpha_i=\alpha^{q^i}$, with indices in $\Z/d\Z$. Thus $S=\{\alpha_0,\ldots,\alpha_{d-1}\}$ is a single Frobenius cycle, $e=1$, and all multiplicities equal one. We distinguish $\Gamma$ from its ground-field-valued subgroup
\[
 \Gamma(k)=\Gamma\cap(k^\times\times k).
\]
Invariance under prescribed ground-field transformations is the setting of \cite{ReisGL}; the geometric group considered here also allows parameters in the splitting field.

\subsubsection{Rational points and translation subgroups}
\begin{proposition}\label{prop:k-points}
The group $\Gamma(k)$ is cyclic of order dividing $d$. More precisely,
\[
 J_\alpha=\{j\in\Z/d\Z:\alpha^{q^j}\in k+k\alpha\}
\]
is a subgroup of $\Z/d\Z$, and $(a,b)\mapsto j$, where $a\alpha+b=\alpha^{q^j}$, is an isomorphism $\Gamma(k)\simeq J_\alpha$. In particular, $|V\cap k|\le p$.
\end{proposition}
\begin{proof}
Every ground-field affine map commutes with Frobenius. The centralizer of a $d$-cycle in the permutation group on $S$ is the cyclic group generated by that cycle. The affine action is faithful because $|S|\ge2$. Conversely, if $\alpha^{q^j}=a\alpha+b$ with $a,b\in k$, then $a\ne0$ and the same identity holds at every conjugate root. It therefore gives an element of $\Gamma(k)$. An elementary abelian $p$-subgroup of a cyclic group has order at most $p$.
\end{proof}

\begin{theorem}\label{thm:translation-size}
Let $Q=|V|$. Then $Q\mid d$ and
\begin{equation}\label{eq:small-V}
 Q\in\{1,p\}\quad\text{if }p\text{ is odd},\qquad
 Q\in\{1,2,4\}\quad\text{if }p=2.
\end{equation}
If $p=2$ and $Q=4$, then $d=4t$ with $t$ odd, $V\subseteq\F_{q^2}$, $|V\cap k|=2$, and Frobenius acts on $V$ as a nonidentity operator of order two.
\end{theorem}
\begin{proof}
Write $Q=p^r$ and $t=d/Q$. Translations act freely on $S$, and Frobenius acts transitively on $S/V$, hence as a $t$-cycle. By \cref{lem:return}, $F^t$ on the fiber $\alpha+V$ is a single $p^r$-cycle. In the coordinate $\alpha+v\leftrightarrow v$, it is the affine map
\[
 T(v)=F^t(v)+(\alpha_t-\alpha).
\]
Embed $\AGL_r(\F_p)$ in $\GL_{r+1}(\F_p)$ by homogeneous coordinates. The matrix of $T$ has order $p^r$, so it is unipotent, say $I+N$, with $N^{r+1}=0$. Consequently
\[
 p^r=\ord(T)\le p^{\lceil\log_p(r+1)\rceil}.
\]
For $p\ge3$ and $r\ge2$, one has $p^{r-1}\ge r+1$, making this inequality impossible. The same argument excludes $p=2$, $r\ge3$. This proves \eqref{eq:small-V}.

In the remaining case $p=2$, $r=2$, the affine return has order four. Its linear part is not the identity, since a translation has order at most two. Put $A=F|_V$. Every element of $\GL_2(\F_2)$ has order $1$, $2$ or $3$, so $A^t$ having order two forces $A$ to have order two and $t$ to be odd. Thus $V\subseteq\F_{q^2}$. A nonidentity involution on a two-dimensional space in characteristic two has a one-dimensional fixed space, giving $|V\cap k|=2$.
\end{proof}

\begin{theorem}\label{thm:differences}
If $p\nmid d$, then $V=0$. If $p\mid d$, let
\[
 b_p=\alpha^{q^{d/p}}-\alpha\ne0.
\]
Then
\begin{equation}\label{eq:line-test}
 V\ne0\quad\Longleftrightarrow\quad
 b_p^{q-1}\in\F_p^\times
 \quad\Longleftrightarrow\quad b_p^{(q-1)(p-1)}=1.
\end{equation}
In odd characteristic, this gives $V=\F_pb_p$. In characteristic two, a four-element group is possible only when $d=4t$ with $t$ odd. In that case put
\[
 b=\alpha^{q^{2t}}+\alpha,\qquad u=\alpha^{q^t}+\alpha.
\]
Then
\begin{equation}\label{eq:plane-test}
 Q=4\quad\Longleftrightarrow\quad b^q=b\text{ and }u^q=u+b,
\end{equation}
and $V=\{0,b,u,b+u\}$. If this test fails, \eqref{eq:line-test} gives either $V=\{0,b_2\}$ or $V=0$.
\end{theorem}
\begin{proof}
If $|V|=p$, the return $F^{d/p}$ on $\alpha+V$ is a $p$-cycle. Its linear part has order both prime to $p$ and dividing $p$, hence is the identity. The return is a nonzero translation, so $V=\F_pb_p$, and Frobenius stability is exactly \eqref{eq:line-test}.

Conversely, suppose $F(b_p)=\lambda b_p$ with $\lambda\in\F_p^\times$. Since $b_p\in K$, $\lambda^d=1$. The order of $\lambda$ is prime to $p$, so $\lambda^{d/p}=1$. It follows that
\[
 \alpha^{q^{j d/p}}=\alpha+jb_p\qquad(0\le j<p).
\]
Applying further powers of Frobenius shows that all roots form a union of $\F_pb_p$-cosets. Thus $\F_pb_p\subseteq V$. This settles odd characteristic and the two-element case.

If $Q=4$, \cref{thm:translation-size} gives $t$ odd and $A=F|_V$ of order two. With $u=\alpha_t+\alpha$, the square of the return gives $b=u+A(u)$; hence $A(b)=b$ and \eqref{eq:plane-test} follows. Conversely, under \eqref{eq:plane-test}, $b\ne0$ and $u\notin\{0,b\}$. The space $W=\{0,b,u,b+u\}$ is Frobenius-stable and $F^t(u)=u+b$. The four roots in the return fiber are
\[
 \alpha,\quad\alpha+u,\quad\alpha+b,\quad\alpha+u+b.
\]
Their Frobenius translates exhaust $S$, so $W\subseteq V$; the size bound makes it equal to $V$. This case also satisfies \eqref{eq:line-test}, proving the full statement.
\end{proof}
If the roots are indexed by a cyclotomic coset, $\alpha_i$ is obtained by multiplying its index by $q^i$, and each difference can be evaluated in the fixed exponent model.

\subsubsection{Determination of the affine stabilizer}
Set $L=L_V$, $\beta=L(\alpha)$ and $D=d/Q$.
\begin{lemma}\label{lem:irr-quotient}
The element $\beta$ is nonzero of degree $D$ over $k$. Its minimal polynomial $G$ is irreducible and
\[
 f(X)=G(L(X)).
\]
The root set $\mathcal B=\{\beta,\beta^q,\ldots,\beta^{q^{D-1}}\}$ of $G$ has no nonzero geometric translation stabilizer.
\end{lemma}
\begin{proof}
If $\alpha\in V$, translation by $-\alpha$ would send a root to zero. Hence $\beta\ne0$. The fibers of $L$ on $S$ are precisely the $V$-cosets, and Frobenius cycles through all $D$ fibers, proving the degree and product identity. If $\mathcal B+t=\mathcal B$, choose $\alpha'\in S$ with $L(\alpha')=\beta+t$ and put $b=\alpha'-\alpha$. Then $L(b)=t$, and $L(S+b)=\mathcal B$. Every fiber over $\mathcal B$ is fully contained in $S$, so $S+b=S$. Maximality of $V$ gives $b\in V$, hence $t=0$.
\end{proof}
Define the scalar normalizer
\[
 N_V=\{a\in K^\times:aV=V\},\qquad\eta=|N_V|.
\]
It is $K^\times$ when $V=0$ and $\F_p^\times$ when $|V|=p$. If $p=2$, $|V|=4$, and
\[
 L=X^4+a_2X^2+a_1X,\qquad a_1\ne0,
\]
then $\eta=3$ when $a_2=0$, and $\eta=1$ otherwise. Indeed, $L(aX)=aL(X)$ gives $a^3=1$ and $a_2(a^2-a)=0$. In the first case $V=b\F_4$ for any nonzero $b\in V$.

For a positive integer $n$ prime to $q$, put $d_n=\ord_n(q)$, with $\ord_1(q)=1$, and define
\begin{equation}\label{eq:sq}
 s_q(n)=\max\{u:u\mid\gcd(n,d_n),\quad
                       \ord_{n/u}(q)=d_n/u\}.
\end{equation}
This is the support period of any monic irreducible polynomial whose nonzero roots have order $n$. To see this, reduce a primitive cyclotomic orbit modulo $n/u$. Its size becomes $\ord_{n/u}(q)$ and its fibers have equal cardinality. It is a union of full cosets of the subgroup of order $u$ exactly when $d_n=u\ord_{n/u}(q)$. These are exactly the allowable subgroup orders in the MED stabilizer, and \cref{thm:period} gives \eqref{eq:sq}.

\begin{theorem}\label{thm:reconstruct}
With the preceding notation, the full group is determined as follows.
\begin{enumerate}
\item If $D=1$, then $f=L-\beta$ and
\[
 \Gamma=\{(a,(1-a)\alpha+v):a\in N_V,\ v\in V\},
 \qquad |\Gamma|=Q\eta.
\]
\item If $D>1$, for each $z\in k$ put
\[
 n_z=\ord_{K^\times}(\beta-z),\qquad
 m_z=\gcd(\eta,s_q(n_z)).
\]
At most one $z$ has $m_z>1$. If there is none, $\Gamma=\{(1,v):v\in V\}$. Otherwise, write $z_0$ for that value and $m=m_{z_0}$. Then
\begin{equation}\label{eq:reconstructed-group}
 \Gamma=\{(a,b)\in K^\times\times K:a^m=1,
                                  \ L(b)=(1-a)z_0\},
 \qquad |\Gamma|=Qm.
\end{equation}
There is a monic irreducible $H\in k[Y]$ such that
\[
 f=H((L-z_0)^m),\qquad H(0)\ne0,\qquad\deg H=D/m.
\]
If $p\nmid D$, the only possible center is
\[
 z_0=D^{-1}\Tr_{\F_{q^D}/k}(\beta).
\]
\end{enumerate}
\end{theorem}
\begin{proof}
Every $(a,b)\in\Gamma$ normalizes $V$, and it acts on the quotient roots by
\[
 L(ax+b)=aL(x)+L(b).
\]
The kernel of the quotient action is $V$. If $D=1$, preserving $\alpha+V$ gives the first assertion immediately.

Suppose $D>1$. By \cref{lem:irr-quotient}, the quotient group has no translations. A nontrivial finite affine group with trivial translation kernel is cyclic with a unique center; Frobenius stability puts this center in $k$. For a fixed $z\in k$, the multiplicative stabilizer of $\mathcal B-z$ has order $s_q(n_z)$ by \eqref{eq:sq}. Intersecting it with $N_V$ gives order $m_z$.

Every such quotient transformation lifts. If $a\in N_V$ preserves $\mathcal B$ with center $z$, choose $\alpha'\in S$ with $L(\alpha')=a\beta+(1-a)z$. Then $b=\alpha'-a\alpha$ satisfies $L(b)=(1-a)z$, and $x\mapsto ax+b$ preserves all fibers above $\mathcal B$. Its other lifts are obtained by adding elements of $V$. Hence a value $m_z>1$ gives a nonidentity quotient element, whose center must be the unique one. This proves uniqueness and \eqref{eq:reconstructed-group}.

The center is not a root of $G$, since $D>1$ and $G$ is irreducible. Its root orbits under the cyclic quotient have size $m$, so $G(Y)=H((Y-z_0)^m)$. A factorization of $H$ would factor $G$, proving irreducibility. Finally, comparing the sum of quotient roots under a nonidentity slope $a$ gives
\[
 \Tr(\beta)=a\Tr(\beta)+D(1-a)z_0,
\]
which yields the stated trace formula.
\end{proof}
The integers $n_z$ in the criterion depend on the root $\beta$ and the center $z$.

\begin{corollary}\label{cor:group-shapes}
If $Q=1$, $\Gamma$ is cyclic of order $m\mid d$, prime to $p$. If $Q=p$, then $\Gamma\simeq C_p\rtimes C_m$ with $m\mid p-1$; moreover $m\mid D$ for $D>1$, and $m=p-1$ for $D=1$. If $Q=4$ in characteristic two, then $\Gamma$ is either $C_2\times C_2$ or the alternating group $A_4$. The latter case is possible only when $d=4$ and $q=2^a$ with $a$ odd.
\end{corollary}
\begin{proof}
Only the last assertion needs further argument. A nontrivial slope group has order three, makes $V=b\F_4$, and permutes its three nonzero elements cyclically, giving $A_4$. Frobenius has order two on $V$ by \cref{thm:translation-size}. If $a$ were even, its action on $b\F_4$, after division by $b$, would be multiplication by an element of $\F_4^\times$, of order one or three. Thus $a$ is odd.

If $D>1$, \cref{thm:reconstruct} gives $D=3j$ and $G(Y)=H((Y-z_0)^3)$ with $\deg H=j$. The integer $D$ is odd by \cref{thm:translation-size}, hence $j$ is odd. For a root $\theta\in\F_{q^j}$ of $H$, the cube map on $\F_{q^j}^\times$ is bijective, because $q^j\equiv2\pmod3$. A cube root of $\theta$ would give a root of $G$ of degree at most $j$, contradicting its degree $3j$. Therefore $D=1$ and $d=4$.
\end{proof}

\subsubsection{Polynomials of prime degree}
\begin{theorem}\label{thm:degree-p}
For monic irreducible $f$ of degree $p$, the geometric affine group is nontrivial if and only if
\begin{equation}\label{eq:AS-polynomials}
 f(X)=X^p-b^{p-1}X-c,\qquad b\in k^\times,\ c\in k,
 \qquad\Tr_{k/\F_p}(c/b^p)\ne0.
\end{equation}
Then $V=b\F_p$, $\Gamma\simeq\AGL_1(\F_p)$, $|\Gamma|=p(p-1)$, and $\Gamma(k)$ is the translation group of order $p$. Also $\sper(f)=1$ and $\wt(f)=3$. The number of these monic irreducible polynomials is $q(q-1)/p$.
\end{theorem}
\begin{proof}
If $V=0$, the slope order divides $p$ and is prime to $p$, so $\Gamma$ is trivial. Otherwise $Q=p$ and $S$ is a single affine $\F_p$-line. Frobenius acts on it as a $p$-cycle, so its linear part is the identity; thus a generator $b$ of the translation line lies in $k^\times$. Its orbit polynomial has the form in \eqref{eq:AS-polynomials}.

Writing $y=\alpha/b$ and $a_0=c/b^p$, the equation is $y^p-y=a_0$. Iterating gives
\[
 y^q-y=\Tr_{k/\F_p}(a_0).
\]
A nonzero right-hand side makes Frobenius a $p$-cycle on the roots; a zero right-hand side puts every root in $k$. This proves the irreducibility criterion. The geometric group follows from the single-fiber case of \cref{thm:reconstruct}. Its ground-field group contains $p$ translations and has order dividing $p$ by \cref{prop:k-points}. All three displayed coefficients are nonzero.

There are $(q-1)/(p-1)$ lines $b\F_p$ in $k$, and for each line there are $q(p-1)/p$ admissible constants $c$. Representatives of the same line give the same coefficient $b^{p-1}$ and an equivalent trace test. Multiplication gives $q(q-1)/p$ distinct polynomials.
\end{proof}

\begin{theorem}\label{thm:degree-ell}
Let $d=\ell\ne p$ be prime. A monic irreducible $f$ of degree $\ell$ has nontrivial geometric affine group if and only if
\begin{equation}\label{eq:centered-binomial}
 \ell\mid q-1,\qquad
 f(X)=(X-c)^\ell-a_0,\qquad
 c\in k,\quad a_0\in k^\times\setminus(k^\times)^\ell.
\end{equation}
Its full group is $\mu_\ell$ centered at $c$, and all parameters lie in $k$. The number of such polynomials is $q(q-1)(\ell-1)/\ell$ when $\ell\mid q-1$, and zero otherwise. If $c=0$, then $\sper(f)=\ell$ and $\wt(f)=2$. If $c\ne0$ and $\ell=\sum_j\ell_jp^j$, then
\[
 \sper(f)=1,\qquad
 \Supp(f)=\Bigl\{\sum_jt_jp^j:0\le t_j\le\ell_j\Bigr\},\qquad
 \wt(f)=\prod_j(\ell_j+1).
\]
\end{theorem}
\begin{proof}
Here $V=0$ and a nontrivial slope order divides $\ell$, so it equals $\ell$. The unique center lies in $k$ and is not a root. The normal form is therefore \eqref{eq:centered-binomial}. If $\ell\nmid q-1$, the $\ell$th power map on $k^\times$ is bijective and the binomial has a root in $k$. If $\ell\mid q-1$, all $\ell$th roots of unity belong to $k$. For $\beta^\ell=a_0$, the ratio $\beta^q/\beta$ lies in $\mu_\ell$. It equals one precisely when $a_0$ is an $\ell$th power in $k$; otherwise it has order $\ell$, giving an irreducible polynomial. The center is determined by the coefficient of $X^{\ell-1}$, so the count follows from the number of non-$\ell$th powers.

For $c\ne0$, the coefficient of $X$ is nonzero and the constant coefficient cannot vanish, since $a_0$ is not an $\ell$th power. The identity
\[
 (1+T)^\ell=\prod_j(1+T^{p^j})^{\ell_j}
\]
shows that the remaining nonzero binomial coefficients are exactly those specified by the digit condition. No cancellation between digit choices occurs.
\end{proof}

\begin{corollary}\label{cor:quartic}
Suppose $p=2$ and $f=X^4+c_3X^3+c_2X^2+c_1X+c_0$ is irreducible. If $c_3\ne0$, its full geometric affine group is trivial. If $c_3=0$, then $Q=4$ and $c_1c_0\ne0$; the group is $A_4$ for $c_2=0$ and $C_2\times C_2$ for $c_2\ne0$.
\end{corollary}
\begin{proof}
A translation-stable four-point set has total sum zero. Conversely, translating one root of a four-point set with total sum zero to zero produces $\{0,b,u,b+u\}$, an $\F_2$-plane. Thus a nonzero translation exists precisely when $c_3=0$, and then $Q=4$. The derivative and the constant term are nonzero by irreducibility. Use the scalar-normalizer calculation and $D=1$. If $Q=1$, the slope order is an odd divisor of four, hence is one.
\end{proof}

\subsection{Semi-invariants and coefficient supports}\label{sec:support}

We now return to arbitrary root multiplicities. Assume $|S|\ge2$, write $f=h^e$ with $h'\ne0$ and $d=\deg h$, and retain the normal form of \cref{thm:normal}. Put
\[
 D=d/Q,\qquad r\equiv D\pmod m,\quad0\le r<m,\qquad
 t=(D-r)/m.
\]
Define the degree-character space
\[
 W_{\Gamma,d}=\{g\in k[X]_{\le d}:g(aX+b)=a^dg(X)
                                             \text{ for all }(a,b)\in\Gamma\}.
\]
\begin{theorem}\label{thm:ordinary-basis}
One has
\begin{equation}\label{eq:ordinary-basis}
 W_{\Gamma,d}=\bigoplus_{j=0}^{t}kZ^{r+mj}.
\end{equation}
In particular, $h=Z^rH(Z^m)$ with $H\in k[Y]$ monic of degree $t$. More precisely, if $u=(\kappa-r)/m$, then $H(Y)=Y^uH_0(Y)$.
\end{theorem}
\begin{proof}
Translation invariance puts $g$ in $K[L_V]=K[Z]$ by \cref{lem:subspace}. Because $Z$ is monic over $k$, uniqueness of the expression gives coefficients in $k$. Write $g=\sum_jb_jZ^j$. The slope condition says $b_j\ne0$ only if $j\equiv d\pmod m$. For $V\ne0$, $Q\equiv1\pmod m$, and for $V=0$ one has $Q=1$; hence this condition is $j\equiv D\pmod m$. The degree bound is $Qj\le d$. The displayed powers have distinct degrees, proving independence. Finally $\kappa\equiv D\pmod m$ and \cref{thm:normal} gives the expression for $H$.
\end{proof}
For a finite-dimensional polynomial space $W$, define its support envelope
\[
 E(W)=\bigcup_{g\in W}\Supp(g).
\]
It equals the union of the supports of any basis. It is the smallest set of monomial coordinates containing every member of $W$, but need not be the support of a single member.

Set
\begin{align*}
 E_\Gamma&=\bigcup_{j=0}^{t}\Supp_X(Z^{r+mj}),
 &B_\Gamma(f)&=|E_\Gamma|,\\
 E_0&=\bigcup_{j=0}^{d/(Q\sigma)}\Supp_X(L_V^{\sigma j}),
 &B_0(f)&=|E_0|,\qquad
 \bsp(f)=\tau/s+1.
\end{align*}
The endpoint $d/(Q\sigma)$ is an integer by \cref{prop:origin}.
\begin{theorem}\label{thm:envelopes}
One has
\begin{equation}\label{eq:envelope-chain}
 \Supp(f)\subseteq eE_\Gamma\subseteq eE_0
        \subseteq\{0,s,2s,\ldots,\tau\},
\end{equation}
and therefore
\[
 \wt(f)\le B_\Gamma(f)\le B_0(f)\le\bsp(f).
\]
If the actual distinguished-orbit multiplicity $\kappa$ is used, the first envelope can be replaced by
\[
 E_{\Gamma,\kappa}=\bigcup_{j=0}^{r_0}\Supp_X(Z^{\kappa+mj})
                         \subseteq E_\Gamma.
\]
\end{theorem}
\begin{proof}
The first containment follows from $h\in W_{\Gamma,d}$ and $\Supp(f)=e\Supp(h)$. On the origin-scaling subgroup $M_0$, the character $a^d$ is trivial because $h(0)\ne0$. Thus every element of $W_{\Gamma,d}$ is invariant under $V\rtimes M_0$. By \cref{lem:subspace}, the invariant ring is $K[L_V^\sigma]$, with descent to $k$, proving $E_\Gamma\subseteq E_0$. The powers of $L_V^\sigma$ are invariant under $\mu_\sigma$ acting at the original origin, so their supports have only multiples of $\sigma$. Multiplication of support indices by $e$ gives \eqref{eq:envelope-chain}. The final refinement follows directly from \eqref{eq:normal}, since $\kappa=r+mu$.
\end{proof}
The supports in these formulas are computed after expanding the powers of $Z$ in $X$.

\begin{proposition}\label{prop:digits}
If $V=b\F_p$ and $a_0=b^{p-1}$, then for $j\ge0$,
\[
 \Supp((X^p-a_0X)^j)
   =\{j+(p-1)i:0\le i\le j,\ i\preceq_p j\}.
\]
Here $i\preceq_pj$ means that every base-$p$ digit of $i$ is at most the corresponding digit of $j$. If $j=\sum_\nu j_\nu p^\nu$, the weight is $\prod_\nu(j_\nu+1)$.
\end{proposition}
\begin{proof}
Expand as $\sum_i\binom ji(-a_0)^{j-i}X^{j+(p-1)i}$. The exponents are distinct, and the binomial coefficient is nonzero precisely under the digit condition, by the same factorization of $(1+T)^j$ used in \cref{thm:degree-ell}.
\end{proof}
This formula applies to the binomial $L_V$, not automatically to a centered polynomial $Z=L_V-\delta$ with an extra constant term.

\begin{example}\label{ex:nonzero-center}
Over $\F_3$, let
\[
 f=(X-1)^4-2=X^4-X^3-X-1.
\]
It is squarefree, and its derivative has the unique zero $1$. Any affine stabilizer must fix that point, so the full group is $\{(a,1-a):a^4=1\}$. Thus $V=0$, $m=4$, $\sigma=1$, and $Z=X-1$. The semi-invariant basis is $1,Z^4$, with support union $\{0,1,3,4\}$. Hence
\[
 \wt(f)=B_\Gamma(f)=4<5=\bsp(f).
\]

\end{example}
\begin{example}\label{ex:unattained}
Over $\F_2$, put $Z=X+1$ and
\[
 f=Z^6+Z^3+1=X^6+X^4+X^3+X+1.
\]
The derivative is $Z^2$, so the full geometric group fixes $1$; comparison in the $Z$-coordinate gives slope group $\mu_3$. Thus
\[
 W_{\Gamma,6}=\spn_{\F_2}\{1,Z^3,Z^6\},\qquad
 E_\Gamma=\{0,1,2,3,4,6\}.
\]
A monic candidate is $G=Z^6+uZ^3+w$. The condition $G(0)\ne0$ forces $w=u$, and $G'\ne0$ forces $u=1$. Hence the only candidate satisfying both conditions is $f$, of weight five, although the envelope has size six.
\end{example}

For any specified finite-field semi-invariant family, each original coefficient is an affine linear function of the free coefficients in \eqref{eq:ordinary-basis}. Exact support patterns can therefore be tested by linear equations and nonvanishing conditions, using finite enumeration or inclusion--exclusion. An element attains the full envelope precisely when all of these coordinates are simultaneously nonzero.

\section{The affine stabilizer group scheme of a polynomial}\label{sec:schemes}
In this section $f$ again has arbitrary root multiplicities. We define its affine stabilizer over commutative $k$-algebras, determine its local structure, and compute the associated semi-invariant spaces.

\subsection{Definition and coefficient equations}
\begin{definition}\label{def:affine-scheme}
For a commutative $k$-algebra $R$, set
\[
 \Af(R)=\{(a,b)\in R^\times\times R:f(aX+b)=a^\tau f(X)\}.
\]
The affine stabilizer group scheme of $f$ is the group functor $\Af$, with group law induced by composition of affine transformations.
\end{definition}

Write $f(X)=\sum_{i=0}^{\tau}c_iX^i$, with $c_\tau=1$ and $c_0\ne0$. For $0\le j<\tau$, set
\begin{equation}\label{eq:Ej}
 E_{f,j}(A,B)=A^j\sum_{i=j}^{\tau}\binom ijc_iB^{i-j}-A^\tau c_j
       \ \in k[A,A^{-1},B].
\end{equation}
For any field-valued pair $(a,b)$, these equations vanish if and only if $f(aX+b)=a^\tau f(X)$, and hence if and only if the corresponding affine index permutation preserves $\widetilde\Omega$.

\begin{theorem}\label{thm:scheme}
The functor $\Af$ is a closed subgroup scheme of $\Aff_1$, with coordinate algebra
\begin{equation}\label{eq:coordinate-algebra}
 \Bf=k[A,A^{-1},B]/\If,\qquad
 \If=(E_{f,0},\ldots,E_{f,\tau-1}).
\end{equation}
It is both the stabilizer of the embedded root scheme $\Spec k[X]/(f)\subseteq\mathbb A^1_k$ and the stabilizer of the coefficient line $kf$. If $f=h^e$, then, in the same Laurent polynomial ring,
\begin{equation}\label{eq:frob-ideal}
 \If=\Ih^{[e]}:=(g^e:g\in\Ih).
\end{equation}
This is the Frobenius power of the ideal, not its ordinary product power.
\end{theorem}
\begin{proof}
Expanding $f(AX+B)-A^\tau f(X)$ gives exactly \eqref{eq:Ej}; the leading coefficient equation is automatic. The identities are closed under composition and inverse, which proves the group-scheme assertion.

An affine substitution preserves the ideal $(f)$ precisely when $f(aX+b)=u(X)f(X)$ for some $u(X)\in R[X]$. Since $f$ is monic, multiplication by $f$ cannot cancel the highest nonzero coefficient over any algebra $R$, even one with zero divisors. Equal degrees force $u$ to be constant, and the leading coefficients give $u=a^\tau$. The same comparison identifies the stabilizer of the coefficient line. Finally,
\[
 f(AX+B)-A^\tau f(X)
       =(h(AX+B)-A^d h(X))^e,
\]
whose coefficients at indices $ej$ are $E_{h,j}^e$ and whose other coefficients vanish. This proves \eqref{eq:frob-ideal}.
\end{proof}

\subsection{Translation and scaling subgroups}
The pure-translation equations are
\begin{equation}\label{eq:Dj}
 D_{j,f}(B)=[X^j](f(X+B)-f(X))
           =\sum_{i=j+1}^{\tau}\binom ijc_iB^{i-j}.
\end{equation}
They are not all zero, since $D_{0,f}=f(B)-f(0)$ is nonconstant. Let $\Psi_f$ be their monic greatest common divisor.

\begin{theorem}\label{thm:translation-ideal}
With $f=h^e$ and $V$ as above,
\begin{equation}\label{eq:Psi}
 (D_{0,f},\ldots,D_{\tau-1,f})=(\Psi_f),\qquad
 \Psi_f(B)=L_V(B)^e.
\end{equation}
Thus $\mult_0(\Psi_f)=e$ and $\deg\Psi_f=eQ$. In particular, $(\sper(f),\Psi_f)$ recovers $e$, $V$ and $\sigma$, but need not recover the center and all slopes of $\Gamma$.
\end{theorem}
\begin{proof}
For $h$, the common zeros of the $D_{j,h}$ are exactly the stable translations $V$. Since $h'\ne0$, there is $i>0$ with $i[X^i]h\ne0$. Then
\[
 D_{i-1,h}'(0)=i[X^i]h\ne0,
\]
where the prime here is differentiation with respect to $B$. Hence the common gcd has a simple zero at zero. For $b\in V$,
$D_{j,h}(B+b)=D_{j,h}(B)$,
so it has a simple zero at every common root. It is therefore $L_V$.

Now
\[
 f(X+B)-f(X)=(h(X+B)-h(X))^e.
\]
Its nonzero $X$-coefficients are precisely the $e$th powers of the nonzero difference coefficients of $h$, and their monic gcd is $L_V^e$. The assertions about degree and the zero multiplicity follow from $L_V'(0)\ne0$; \cref{thm:period} then recovers $\sigma$.
\end{proof}

\begin{corollary}\label{cor:sections}
The pure-translation and origin-preserving scaling subgroup schemes are
\begin{equation}\label{eq:two-sections}
 \mathcal U_f=\ker(L_V^e:\Ga\to\Ga),\qquad
 \Af\cap\Gm=\mu_s=\mu_{e\sigma}.
\end{equation}
\end{corollary}
\begin{proof}
The first identity follows from \cref{thm:translation-ideal}, since the gcd is a B\'ezout combination of the difference coefficients. Such an identity remains valid after evaluation in any $R$. For scaling, the nonzero constant coefficient forces $a^\tau=1$, and the other conditions are $a^j=1$ for all positive support indices. In the Laurent polynomial ring, the elements $A^j-1$ generate $A^s-1$ by the Euclidean algorithm on exponents.
\end{proof}
The group scheme $\mu_s$ has rank $s$, although it has only $\sigma$ geometric points.

\subsection{The local algebras of the roots}
Over the splitting field, the defining function gives the coordinate-marked local algebra
\begin{equation}\label{eq:root-local}
 K[X]/(f)\simeq\prod_{\gamma\in T}K[U]/(U^{\Omega(\gamma)}),
 \qquad X\longmapsto(\zeta_M^\gamma+U)_\gamma.
\end{equation}
Consequently, over any commutative $K$-algebra $R$, a pair $(a,b)$ belongs to $\Af(R)$ if and only if
\begin{equation}\label{eq:local-test}
 \prod_{\eta\in T}
       \bigl(a(\zeta_M^\gamma+U)+b-\zeta_M^\eta\bigr)^{\Omega(\eta)}
       \equiv0\pmod{U^{\Omega(\gamma)}}
 \quad(\gamma\in T).
\end{equation}
Indeed, the root factors for different $\gamma$ are comaximal over $R$, because their root differences are units. These congruences are equivalent to divisibility by $f$, and the same monicity argument gives the required scalar identity.

\subsection{The Lie algebra}

\begin{theorem}\label{thm:lie}
Over a splitting field $K$,
\begin{align}
 \Lie(\Af)_K
 &=\{(a,b)\in K^2:(aX+b)f'(X)=\tau af(X)\}\label{eq:lie-coeff}\\
 &=\{(a,b)\in K^2:\Omega_K(\alpha)(a\alpha+b)=0
                                          \text{ for every }\alpha\in S\}.
 \label{eq:lie-roots}
\end{align}
Here the integer multiplicities in \eqref{eq:lie-roots} are read modulo $p$. If $n_* = |\{\alpha\in S:p\nmid\Omega_K(\alpha)\}|$, then
\[
 \dim_k\Lie(\Af)=
 \begin{cases}2,&n_*=0,\\1,&n_*=1,\\0,&n_*\ge2.\end{cases}
\]
\end{theorem}
\begin{proof}
Substitute $A=1+\varepsilon a$, $B=\varepsilon b$ into \cref{thm:scheme}, where $\varepsilon^2=0$. The coefficient of $\varepsilon$ gives \eqref{eq:lie-coeff}. In $K(X)$,
\[
 (aX+b)\frac{f'}f-\tau a
   =\sum_{\alpha\in S}
       \frac{\Omega_K(\alpha)(a\alpha+b)}{X-\alpha}.
\]
The poles are distinct, so this rational function is zero exactly when every numerator is zero. No constraint, one nonzero constraint, or two constraints at different roots give dimensions two, one, or zero, respectively. Extension of scalars preserves the rank of the equations.
\end{proof}

\subsection{Local structure}
Write $f=h^e$, with $h'\ne0$, and define
\[
 n_{*,h}=|\{\alpha\in S:p\nmid\mult_\alpha(h)\}|.
\]
It is positive because $h'\ne0$.

\begin{theorem}\label{thm:local-etale}
Suppose $|S|\ge2$ and $n_{*,h}\ge2$. Then $\Ah$ is finite \emph{\'etale}. If $K$ contains all roots and $\Gamma$ is its geometric affine group, there is an isomorphism of $K$-algebras
\begin{equation}\label{eq:local-product}
 \Bf\otimes_kK\simeq
       \prod_{(a_0,b_0)\in\Gamma}K[U,W]/(U^e,W^e).
\end{equation}
In particular, $\rank_k\Af=e^2|\Gamma|$.
\end{theorem}
\begin{proof}
The geometric points are a finite set, determined by the images of two distinct roots. A finite-type affine scheme with this property is zero-dimensional and finite over the field. By \cref{thm:lie}, the tangent space of $\Ah$ at the identity is zero. Group translation gives the same at every geometric point. The maximal ideal $\mathfrak m$ of a local component therefore satisfies $\mathfrak m/\mathfrak m^2=0$; Nakayama's lemma gives $\mathfrak m=0$. Thus $\Ah$ is geometrically reduced, hence finite \emph{\'etale} over the perfect field $k$.

All geometric points lie in $K$, so
\[
 \Ih K[A,A^{-1},B]
    =\bigcap_{(a_0,b_0)\in\Gamma}(A-a_0,B-b_0).
\]
These maximal ideals are pairwise comaximal. Frobenius powers preserve products and comaximality, so \eqref{eq:frob-ideal} gives
\[
 \If K[A,A^{-1},B]
    =\bigcap_{(a_0,b_0)\in\Gamma}
          ((A-a_0)^e,(B-b_0)^e).
\]
Apply the Chinese remainder theorem. Each component has basis $U^iW^j$, $0\le i,j<e$, proving the rank formula.
\end{proof}
In the local coordinate algebra at the identity, $A=1+U$, $B=W$, and the comultiplication is
\[
 \Delta U=U\otimes1+1\otimes U+U\otimes U,\qquad
 \Delta W=W\otimes1+1\otimes W+U\otimes W.
\]
Thus the identity component is the Frobenius kernel $\alpha_e\rtimes\mu_e$ of the affine group. In particular, for an irreducible polynomial of degree at least two, $e=1$ and $\Af$ is finite \emph{\'etale}; its geometric group together with Frobenius descent determines the whole group scheme.

\begin{theorem}\label{thm:local-centered}
Suppose $|S|\ge2$ and $n_{*,h}=1$. Let $c_*$ be the unique root whose multiplicity in $h$ is not divisible by $p$. Then $c_*\in k$, $p\nmid d$, and
\begin{equation}\label{eq:unique-root-derivative}
 (X-c_*)h'(X)=dh(X).
\end{equation}
Write $h(X+c_*)=\sum_{j=0}^du_jX^j$ and define
\[
 D_h=\gcd\{d-j:0\le j<d,\ u_j\ne0\}.
\]
This is a positive integer divisible by $p$, and
\begin{equation}\label{eq:unique-root-ideal}
 \If=\bigl((B+(A-1)c_*)^e,\ A^{eD_h}-1\bigr).
\end{equation}
In coordinates centered at $c_*$,
\begin{equation}\label{eq:unique-root-group}
 \Af\simeq\alpha_e\rtimes\mu_{eD_h}.
\end{equation}
If $D_h=p^uD_0$, $p\nmid D_0$, it has $D_0$ geometric points, each with local algebra $K[U,W]/(U^{ep^u},W^e)$. Its rank is $e^2D_h$.
\end{theorem}
\begin{proof}
Frobenius fixes the unique distinguished root $c_*$, so $c_*\in k$. Modulo $p$, the degree of $h$ equals the multiplicity at $c_*$ and is nonzero. The logarithmic derivative is $h'/h=d/(X-c_*)$, giving \eqref{eq:unique-root-derivative}.

For any algebra $R$ and $(a,b)\in\Ah(R)$, differentiate $h(aX+b)=a^dh(X)$ and combine with \eqref{eq:unique-root-derivative}. The result is
\[
 (b+(a-1)c_*)h'(X)=0.
\]
Since the leading coefficient of $h'$ is the unit $d$, this forces $b=(1-a)c_*$ even if $R$ has zero divisors. After centering, \eqref{eq:unique-root-derivative} says that $u_j\ne0$ implies $j\equiv d\pmod p$. There is a nonzero coefficient below degree $d$, because $h$ has at least two roots, so $D_h$ is defined and divisible by $p$.

For a center-fixing scaling, the remaining equations are $a^{d-j}=1$ whenever $u_j\ne0$. The exponent Euclidean algorithm gives
\[
 \Ih=(B+(A-1)c_*,A^{D_h}-1).
\]
Taking its Frobenius power gives \eqref{eq:unique-root-ideal}. If $w=b+(a-1)c_*$, the centered map is $Y\mapsto aY+w$, with $w^e=0$ and $a^{eD_h}=1$. Its composition law is $(w,a)(w',a')=(w+aw',aa')$, proving \eqref{eq:unique-root-group}. Finally,
$A^{eD_h}-1=(A^{D_0}-1)^{ep^u}$,
and $A^{D_0}-1$ has distinct roots. This gives the local algebras and ranks.
\end{proof}

\begin{proposition}\label{prop:single-root}
If $S=\{c\}$, then $c\in k^\times$ and $f=(X-c)^\tau$. Write $\tau=et$, $p\nmid t$. In centered coordinates,
\[
 \Af\simeq\alpha_e\rtimes\Gm.
\]
If $\tau=\sum_j\tau_jp^j$, $0\le\tau_j<p$, then
\[
 \Supp(f)=\Bigl\{\sum_jb_jp^j:0\le b_j\le\tau_j\Bigr\},\qquad
 \wt(f)=\prod_j(\tau_j+1).
\]
\end{proposition}
\begin{proof}
The unique root is Frobenius-fixed. For $Y=X-c$, the condition on $(a,w)$ is $(aY+w)^{et}=a^{et}Y^{et}$. The coefficient of $Y^{e(t-1)}$ forces $w^e=0$ since $ta^{e(t-1)}$ is a unit. This condition is also sufficient. For the support, factor
\[
 (X-c)^\tau=\prod_j(X^{p^j}-c^{p^j})^{\tau_j}.
\]
The small binomial coefficients are nonzero, and different digit choices give different exponents.
\end{proof}
In this case $\dim\Af=1$.

\subsection{Invariant rings and finite flat quotients}

Using the subgroup schemes in \cref{cor:sections}, let
\[
 \Psi=L_V^e,\qquad \mathcal G_0=\mathcal U_f\rtimes\mu_\sigma,
 \qquad P_0=\Psi^\sigma.
\]
Then $\mathcal G_0$ is a subgroup scheme of $\Af$ whose invariants are described as follows.

\begin{proposition}\label{prop:torsor}
The invariant rings are
\[
 k[X]^{\mathcal U_f}=k[\Psi],\qquad k[X]^{\mathcal G_0}=k[P_0].
\]
On $U=\Spec k[X,1/L_V]$, the map $X\mapsto P_0(X)$ is a $\mathcal G_0$-torsor over $\Gm$ of rank $eQ\sigma$. It is \emph{\'etale} when $e=1$, and has an inseparable part when $e>1$.
\end{proposition}
\begin{proof}
Divide an invariant polynomial by the monic additive polynomial $\Psi$. In $k[B]/(\Psi(B))$, substitution $X\mapsto X+B$ and uniqueness of quotient and remainder show that both are invariant. If the remainder is $R$ with $\deg R<\deg\Psi$, evaluating at $X=0$ gives $\Psi(B)\mid R(B)-R(0)$, so $R$ is constant. Induction proves the first invariant-ring formula. Since $\Psi(aX)=a^e\Psi(X)$ for $a\in\mu_\sigma$, and $(e,\sigma)=1$, taking scaling invariants gives the second.

Monicity of $P_0$ makes the coordinate ring of $U$ finite free of rank $eQ\sigma$ over $k[P_0,P_0^{-1}]$. To check the torsor identity, let $x,y$ be points over an arbitrary algebra with invertible $L_V$-values and $P_0(x)=P_0(y)$. The ratio $\Psi(y)/\Psi(x)$ lies in $\mu_\sigma$, so it is uniquely $a^e$ with $a\in\mu_\sigma$. Set $b=y-ax$. Additivity gives $\Psi(b)=0$. Thus there is a unique $(b,a)\in\mathcal G_0$ taking $x$ to $y$, proving the torsor isomorphism. The derivative is invertible on $U$ exactly in the case $e=1$.
\end{proof}

\subsection{Semi-invariants and coefficient supports}\label{subsec:scheme-support}

Let $\Bf$ be the coordinate algebra in \cref{thm:scheme}, and define
\[
 \Wsch=\{g\in k[X]_{\le\tau}:g(AX+B)=A^\tau g(X)
                                                     \text{ in }\Bf[X]\},
\]
with
\[
 \Es=E(\Wsch),\qquad\bsch(f)=|\Es|.
\]
We use the following coefficient criterion for roots of unity group schemes. For the scaling action of $\mu_n$ on $Y$,
\begin{equation}\label{eq:mu-character}
 \sum_jb_j(aY)^j=a^t\sum_jb_jY^j\text{ on all algebra-valued points}
 \quad\Longleftrightarrow\quad
 b_j\ne0\ \Longrightarrow\ j\equiv t\pmod n.
\end{equation}
Indeed, $1,A,\ldots,A^{n-1}$ form a basis of $k[A,A^{-1}]/(A^n-1)$. The equivalence holds also when $p\mid n$.

\begin{theorem}\label{thm:scheme-basis}
For $|S|\ge2$,
\[
 \Supp(f)\subseteq\Es\subseteq eE_\Gamma,
 \qquad \wt(f)\le\bsch(f)\le B_\Gamma(f).
\]
More precisely:
\begin{enumerate}
\item If $n_{*,h}\ge2$, then
\begin{equation}\label{eq:scheme-etale-basis}
 \Wsch=\{g^e:g\in W_{\Gamma,d}\},\qquad
 \Es=eE_\Gamma,\qquad\bsch(f)=B_\Gamma(f).
\end{equation}
\item If $n_{*,h}=1$, with $c_*$ and $D_h$ from \cref{thm:local-centered}, then
\begin{equation}\label{eq:scheme-centered-basis}
 \Wsch=\spn_k\{(X-c_*)^{ej}:0\le j\le d,\ j\equiv d\pmod{D_h}\}.
\end{equation}
Its support envelope is the union of the original-coordinate supports of these basis elements.
\end{enumerate}
\end{theorem}
\begin{proof}
The polynomial $f$ itself belongs to $\Wsch$. The subgroup $\alpha_e$ of pure translations in $\Af$ forces every member to lie in $k[X^e]$. For an elementary verification, if a nonzero coefficient has index $j$ not divisible by $e$, let $r_1=p^{v_p(j)}<e$. The binomial digit identity gives $\binom{j}{r_1}\ne0\pmod p$. The corresponding coefficient in $\partial^{[r_1]}g$ is nonzero, contradicting invariance modulo $B^e$. Thus $g=G^e$ uniquely over the perfect field $k$. At geometric points one can take $e$th roots of the semi-invariant identity, obtaining $G\in W_{\Gamma,d}$.

If $n_{*,h}\ge2$, the ideal $\Ih$ is geometrically reduced by \cref{thm:local-etale}. For $G\in W_{\Gamma,d}$, the coefficients of its semi-invariant equation vanish at every geometric point, hence belong to $\Ih$. Their $e$th powers belong to $\If=\Ih^{[e]}$, so $G^e\in\Wsch$. This proves \eqref{eq:scheme-etale-basis}. The Frobenius power map is semilinear in general, not $k$-linear, but it maps onto a $k$-space and multiplies support indices by $e$.

In the second case use $Y=X-c_*$. Invariance under $\alpha_e$ restricts the exponents to $ej$, and \eqref{eq:mu-character} for $\mu_{eD_h}$ imposes $ej\equiv ed\pmod{eD_h}$. This is exactly \eqref{eq:scheme-centered-basis}. Each listed element satisfies both conditions and their semidirect product, and distinct degrees give linear independence.
\end{proof}
In case~(i), the two support bounds agree. In case~(ii), the exponent congruence is taken modulo the full integer $D_h$.

\begin{example}\label{ex:hidden}
Over $\F_5$, consider
\begin{align*}
 h_1&=(X-1)(X-2)^5=X^6-X^5-2X+2,\\
 h_2&=(X-1)^4(X-2)(X-3)
       =X^6+X^5+2X^4+2X^3+2X^2+X+1.
\end{align*}
Both have $e=1$ and trivial geometric affine stabilizer. For $h_1$, unequal multiplicities force both roots to be fixed. For $h_2$, the root of multiplicity four forces $a+b=1$, while the weighted root sum gives $4a+6b=4$, hence $a=1$, $b=0$.

Their weights are four and seven, respectively. For $h_1$, $c_*=1$ and $h_1(X+1)=X^6-X$, so $D_{h_1}=5$ and $\mathcal A_{h_1}$ is $\mu_5$ centered at $1$. Its semi-invariant basis in degree at most six is $X-1,(X-1)^6$, and the envelope is $\{0,1,5,6\}$. For $h_2$, three root multiplicities are nonzero modulo five, so its stabilizer scheme is trivial. After taking fifth powers, the respective total ranks are $125$ and $25$, in agreement with the two local formulas.
\end{example}
\section{Applications to constacyclic codes}\label{sec:codes}
\subsection{Affine substitutions}
Let $C=(f)\subseteq R_{N,\lambda}$, where $f$ is monic, $f\mid X^N-\lambda$, and $0<\tau=\deg f<N$. In the fixed polynomial coordinate,
\begin{equation}\label{eq:code-representatives}
 C=\{f(X)q(X):q\in k[X]_{<N-\tau}\}\subseteq k[X]_{<N},
 \qquad \dim_kC=N-\tau.
\end{equation}
Write $N=p^\eta n$ with $p\nmid n$, and let $\lambda_0^{p^\eta}=\lambda$ in $k$. Then
\[
 X^N-\lambda=(X^n-\lambda_0)^{p^\eta}.
\]
For $r_\lambda=\ord(\lambda_0)=\ord(\lambda)$, set $M=nr_\lambda$ and choose $c_\lambda$ with $\lambda_0=\zeta_{r_\lambda}^{c_\lambda}$. The ambient defining set is
\[
 T_{N,\lambda}=\{\gamma\in\Z/M\Z:\gamma\equiv c_\lambda\pmod{r_\lambda}\}.
\]
The generator is equivalently specified by \eqref{eq:defining} with $0\le\Omega\le p^\eta$ and support contained in $T_{N,\lambda}$, excluding the zero function and the whole ambient multiplicity function.

\begin{theorem}\label{thm:code-stabilizer}
For any commutative $k$-algebra $R$, let $C_R=C\otimes_kR\subseteq R[X]_{<N}$, and let $U_{a,b}$ denote substitution $g(X)\mapsto g(aX+b)$, with $a\in R^\times$. Then
\begin{equation}\label{eq:code-stabilizer}
 U_{a,b}(C_R)=C_R
 \quad\Longleftrightarrow\quad
 f(aX+b)=a^\tau f(X)
 \quad\Longleftrightarrow\quad(a,b)\in\Af(R).
\end{equation}
\end{theorem}
\begin{proof}
The semi-invariant identity sends $fq$ to $a^\tau f(X)q(aX+b)$, with the same degree bound on the multiplier, and the inverse substitution proves equality. Conversely, $U_{a,b}f\in C_R$ has degree $\tau$ and unit leading coefficient. Monicity of $f$ and \eqref{eq:code-representatives} force it to be a constant multiple of $f$; the leading coefficient determines that constant as $a^\tau$.
\end{proof}
Here $C$ is viewed as a subspace of polynomials of degree less than $N$. The substitution need not preserve the ideal $(X^N-\lambda)$.

\subsection{Coordinate decompositions}
\begin{theorem}\label{thm:decomposition}
Let $s=\sper(f)=e\sigma$. Then $s\mid N$, $f=\fsp(X^s)$, and $\fsp\mid Y^{N/s}-\lambda$. Put
\[
 \Csp=(\fsp)\subseteq k[Y]/(Y^{N/s}-\lambda).
\]
The interleaver
\begin{equation}\label{eq:interleaver}
 \Phi_s:\Csp^{\oplus s}\longrightarrow C,
 \qquad(c_0,\ldots,c_{s-1})\longmapsto
                   \sum_{j=0}^{s-1}X^jc_j(X^s)
\end{equation}
is a $k$-linear coordinate-permutation isometry. The $s$ coordinate blocks
\[
 R_j=j+s\Z/N\Z\qquad(0\le j<s)
\]
are exactly the indecomposable support blocks of $C$. Every coordinate direct-sum decomposition is obtained by grouping these blocks. In particular, if $\nu(C)$ is the number of indecomposable coordinate components, then
\begin{equation}\label{eq:nu-bridge}
 \nu(C)=\rank_k(\Af\cap\Gm)=\sper(f)=e|\Sigma_\Omega|.
\end{equation}
\end{theorem}
\begin{proof}
Use $k[X]=\bigoplus_{j=0}^{s-1}X^jk[X^s]$. If $X^N-\lambda=\fsp(X^s)q(X)$ and $N\not\equiv0\pmod s$, comparison of the zero residue component would make the nonconstant polynomial $\fsp$ divide $-\lambda$. Thus $s\mid N$, and componentwise division gives $\fsp\mid Y^{N/s}-\lambda$. The map \eqref{eq:interleaver} places coordinate $i$ of component $j$ at coordinate $si+j$. Every coordinate occurs once. Divisibility by $f$ is equivalent to divisibility of every component by $\fsp$, so the image is $C$.

For indecomposability, we use the support argument of \cite[Theorem~8]{KK}, with constacyclic shifts in place of cyclic shifts. The generator $f$ is a support-minimal nonzero codeword: a nonzero codeword supported inside $\Supp(f)$ has degree at most $\tau$, hence, being divisible by $f$, is a scalar multiple of $f$. Every constacyclic shift of $f$ is also support-minimal, since the shift is a monomial isometry.

Connect the coordinates in each shifted support. Since $0\in\Supp(f)$, the connected component of zero is the subgroup generated by the positive support indices in $\Z/N\Z$. As their greatest common divisor is $s\mid N$, that subgroup is $s\Z/N\Z$. Thus the connected components are exactly the $R_j$. The shifted generators span $C$ and each is contained in one block, so these blocks yield a direct sum. In any coordinate direct sum, a support-minimal codeword must lie in one component: otherwise projection onto one nonzero summand produces a strictly smaller nonzero support. Hence no $R_j$ can be split. This proves the indecomposability and grouping statements.

Finally, \cref{cor:sections} identifies the scaling subgroup scheme with $\mu_s$, of rank $s$, while \cref{thm:period} identifies $s=e|\Sigma_\Omega|$. This proves \eqref{eq:nu-bridge}.
\end{proof}
The coordinate block partition is canonical in the fixed polynomial representation. The polynomial $\fsp$ is the monic generator of each compressed component in the induced coordinate.

\begin{corollary}\label{cor:core-bounds}
If $C$ has parameters $[N,k_C,d_H(C)]$, then $\Csp$ has parameters
\[
 [N/s,k_C/s,d_H(C)].
\]
Consequently
\begin{equation}\label{eq:core-bound}
 d_H(C)\le\bsp(C):=\tau/s+1,
\end{equation}
with equality if and only if $\Csp$ is MDS. More generally, every applicable upper bound $U_q(M,K)$ for linear codes yields
\[
 d_H(C)\le U_q(N/s,k_C/s).
\]
\end{corollary}
\begin{proof}
A nonzero direct-sum word has weight at least the least nonzero weight of one component, and this is attained in a single component. Dimensions add. Applying the usual Singleton bound to the core gives \eqref{eq:core-bound} and its equality criterion.
\end{proof}
For simple-root codes, \eqref{eq:core-bound} is the arithmetic Singleton value defined through the coarsest MED representation in \cite{AS}. For repeated-root codes, we use the term \emph{support-period value} $\bsp(C)$.

The usual direct-sum consequences can be recorded together. If $D=\Csp$, $s=\sper(f)$, and $d_0=d_H(D)$, then
\begin{equation}\label{eq:standard-consequences}
 A_{d_0}(C)=sA_{d_0}(D),\quad W_C=W_D^s,\quad
 \rho(C)=s\rho(D),\quad C^\perp\simeq(D^\perp)^{\oplus s},
 \quad\Hull(C)\simeq\Hull(D)^{\oplus s}.
\end{equation}
Here $A_i$ counts words of weight $i$, $W$ is the Hamming weight enumerator, $\rho$ the covering radius, and $\perp$ the Euclidean dual. The interleaver is a pure coordinate permutation, so it preserves both weight and the standard inner product. Residue-class coordinate projections are self-adjoint; the sets of divisors defining projection-stable decompositions of $C$ and $C^\perp$ are therefore equal. Thus $(C^\perp)_{\mathrm{sp}}=D^\perp$.

\subsection{Polynomial substitutions}
We next consider the module decomposition associated with a polynomial substitution, including $P_0$ from \eqref{eq:P0}.

\begin{proposition}\label{prop:module}
Let $P\in k[X]$ be monic of degree $D_P>0$ with $P(0)=0$. Write
\[
 k[X]=\bigoplus_{j=0}^{D_P-1}X^jk[P(X)]
\]
and let $\pi_j^P$ denote the corresponding projections. The following conditions are equivalent:
\[
 f\in k[P],\qquad \pi_j^P(C)\subseteq C\text{ for all }j,\qquad
 C=\bigoplus_{j=0}^{D_P-1}\bigl(C\cap X^jk[P]\bigr).
\]
If $f=H(P)$, with $r_P=\deg H$, the $j$th component has dimension
\begin{equation}\label{eq:module-dim}
 \ell_j=\max\left(0,1+\floor{\frac{N-1-j}{D_P}}-r_P\right).
\end{equation}
\end{proposition}
\begin{proof}
The polynomials $X^jP^a$ have pairwise distinct leading degrees $j+D_Pa$, covering all nonnegative integers. They form a basis, and each projection does not increase degree. If $f=H(P)$, projection commutes with multiplication by $f$, so it preserves $C$. Projection stability is equivalent to the direct-sum assertion. Conversely, $P(0)=0$ gives $\pi_0^P(f)(0)=f(0)\ne0$. This projection is a codeword of degree at most $\tau$, so it equals a scalar multiple of $f$; the constant coefficient makes the scalar one. Thus $f\in k[P]$. In that case the $j$th component is
\[
 X^jH(P)\{A(P):\deg A<\ell_j\},
\]
with the degree restriction exactly \eqref{eq:module-dim}.
\end{proof}
The proposition applies to $P=P_0$, but $D_P$ need not divide $N$ and the dimensions $\ell_j$ need not be equal. For a centered quotient, the additional hypotheses $f\in k[P]$ and $P(0)=0$ are required.

If $T_P$ is the coefficient matrix of the basis $X^jP^a$ below degree $N$, and $\widetilde C=T_P^{-1}C$, the preserved metric is
\[
 w_P(z)=\wt(T_Pz),\qquad
 d_H(C)=\min_{0\ne z\in\widetilde C}w_P(z).
\]
In the same column-coordinate convention, put $G_P=T_P^{\mathsf T}T_P$. Then
\begin{equation}\label{eq:gram}
 C^\perp=T_P^{-\mathsf T}\widetilde C^\perp,
 \qquad
 \Hull(C)=T_P\bigl(\widetilde C\cap G_P^{-1}\widetilde C^\perp\bigr).
\end{equation}
These follow immediately from $\langle T_Px,T_Py\rangle=x^{\mathsf T}G_Py$. When $P=X^s$ and the basis is arranged by blocks, $T_P$ is a permutation matrix and the ordinary Hamming direct sum is recovered. Otherwise, the pulled-back weight generally does not separate across the module components.

A nonzero translation is already nonisometric on the ambient coefficient space, since it sends $X$ to $X+b$. It also cannot preserve the full set of distinct roots of $X^N-\lambda$: that set has size $n$, prime to $p$, whereas every orbit of a nonzero translation has size $p$. It follows that such a translation does not induce an automorphism of $k[X]/(X^N-\lambda)$.

\subsection{Dimensions and supports of subcodes}
For a finite-dimensional test space $W\subseteq k[X]_{<N}$ with basis $w_1,\ldots,w_t$, put
\[
 D_W=C\cap W,\qquad E(W)=\bigcup_{i=1}^t\Supp(w_i).
\]
Let $M_W$ be a matrix over $k$ whose kernel consists of the parameter vectors $z$ for which $f\mid\sum_i z_iw_i$. For example, take the standard remainder columns modulo $f$. Write $\rho_W=\rank_kM_W$. For $0\le j<N$, let
\[
 r_j=([X^j]w_1,\ldots,[X^j]w_t).
\]
\begin{theorem}\label{thm:effective}
One has $\dim D_W=t-\rho_W$ and
\begin{equation}\label{eq:effective}
 E_{\mathrm{eff}}(W):=\Supp(D_W)
  =\{j:r_j\notin\row_k(M_W)\}
  =\left\{j:\rank_k\begin{pmatrix}M_W\\r_j\end{pmatrix}>\rho_W\right\}.
\end{equation}
If $t>\rho_W$, then
\begin{equation}\label{eq:effective-bound}
 d_H(C)\le |E_{\mathrm{eff}}(W)|-t+\rho_W+1
           \le |E(W)|-t+\rho_W+1.
\end{equation}
For $1\le j\le t-\rho_W$, the $j$th generalized Hamming weight satisfies
\[
 d_j(C)\le |E_{\mathrm{eff}}(W)|-(t-\rho_W)+j.
\]
\end{theorem}
\begin{proof}
The parameter map identifies $D_W$ with $\ker M_W$. Coordinate $j$ vanishes on this kernel exactly when $r_j$ belongs to its annihilator, which is $\row_k(M_W)$. This proves the dimension and support identities. After deleting identically zero coordinates, $D_W$ is a code of length $|E_{\mathrm{eff}}(W)|$ and dimension $t-\rho_W$. Setting $t-\rho_W-1$ coordinates to zero leaves a nonzero vector, proving the Singleton estimate. Setting $t-\rho_W-j$ coordinates to zero leaves at least a $j$-dimensional subcode, proving the generalized weight statement. In both cases $D_W\subseteq C$.
\end{proof}
The effective support is independent of the chosen detection matrix, because matrices with the same kernel have the same row space.

The rank can also be computed directly from the defining function. Over $K$, the Hasse detection matrix has rows indexed by $(\gamma,r)$ with $\gamma\in T$, $0\le r<\Omega(\gamma)$, and columns
\begin{equation}\label{eq:hasse-matrix}
 \bigl(\partial^{[r]}w_i(\zeta_M^\gamma)\bigr)_{\gamma,r}.
\end{equation}
By \eqref{eq:hasse}, its kernel is $D_W\otimes_kK$. Its rank is therefore $\rho_W$, and the same row-space support test may be performed over $K$. The rows with $r>0$ impose the higher root-multiplicity conditions.

\subsection{Subcodes defined by semi-invariants}
Assume again that $|S|\ge2$. Starting from \eqref{eq:normal}, define
\[
 A_*(X)=Z(X)^e,\qquad P_*(X)=A_*(X)^m,\qquad
 B_*(Y)=H_0^{[e]}(Y).
\]
Then
\begin{equation}\label{eq:code-normal}
 f(X)=A_*(X)^\kappa B_*(P_*(X)),\qquad
 B_*(0)\ne0,\qquad\deg A_*=eQ.
\end{equation}
Fix $0\le b<m$ and put
\begin{equation}\label{eq:outer-divisor}
 t_b=\max\left(0,\ceil{\frac{\kappa-b}{m}}\right),\qquad
 \mathscr H_b(Y)=Y^{t_b}B_*(Y),\qquad r_b=\deg\mathscr H_b.
\end{equation}
For a finite nonempty $J\subseteq\Z_{\ge0}$ satisfying
\begin{equation}\label{eq:J-admissible}
 eQ(b+m\max J)<N,
\end{equation}
define
\[
 W_{b,J}=\spn_k\{A_*^{b+mj}:j\in J\},\qquad
 E_{b,J}=\bigcup_{j\in J}\Supp_X(A_*^{b+mj}).
\]
These are semi-invariant powers in the original coordinate, all of character $a^{eb}$ on $\Gamma$. Their degrees are distinct, so $\dim W_{b,J}=|J|$. Let
\[
 \rho_{b,J}=\rank_k\left(\spn_k\{Y^j:j\in J\}
                   \longrightarrow k[Y]/(\mathscr H_b)\right).
\]
If $\mathscr H_b=1$, the target and the rank are zero.

\begin{theorem}\label{thm:outer-subcode}
For every center and every root-multiplicity pattern in \eqref{eq:code-normal},
\begin{equation}\label{eq:outer-equivalence}
 f\mid A_*^bR(P_*)\quad\Longleftrightarrow\quad\mathscr H_b\mid R.
\end{equation}
Consequently
\[
 \dim(C\cap W_{b,J})=|J|-\rho_{b,J},\qquad0\le\rho_{b,J}\le r_b.
\]
When $|J|>\rho_{b,J}$,
\begin{equation}\label{eq:outer-bound}
 d_H(C)\le |E_{\mathrm{eff}}(W_{b,J})|-|J|+\rho_{b,J}+1
            \le |E_{b,J}|-|J|+\rho_{b,J}+1.
\end{equation}
In particular, for $|J|>r_b$ a rank-free bound is $|E_{b,J}|-|J|+r_b+1$.
\end{theorem}
\begin{proof}
For nonzero $R$, write $R=Y^uR_0$ with $R_0(0)\ne0$. At a zero of $A_*$, the factor $R_0(P_*)$ is a unit, as is $B_*(P_*)$. The required multiplicity from $A_*^\kappa$ is thus exactly the condition $b+mu\ge\kappa$, or $Y^{t_b}\mid R$.

The remaining factor is coprime to $A_*$, so it requires $B_*(P_*)\mid R(P_*)$. This is equivalent to $B_*\mid R$: division by the monic $B_*$ leaves a remainder whose nonzero substitution into $P_*$ has smaller degree than $B_*(P_*)$. Since $Y$ and $B_*$ are coprime, the two conditions combine into \eqref{eq:outer-equivalence}. The zero polynomial causes no exception. No squarefreeness is used. The rank-nullity formula and \cref{thm:effective} now prove the remaining assertions.
\end{proof}
The factor $Y^{t_b}$ imposes the required multiplicity on the orbit $c+V$.

The outer rank has a direct root-side matrix. In the notation of \eqref{eq:full-quotient}, the roots of $B_*$ are $\xi^{ez}$, with multiplicities $\Omega_{\mathrm{full}}(z)$. For every such $z$ and $0\le a<\Omega_{\mathrm{full}}(z)$, use the row
\[
 \left(\binom ja\xi^{ez(j-a)}\right)_{j\in J},
\]
with entries zero when $j<a$, and add the $t_b$ rows $([Y^a]Y^j)_{j\in J}$ for $0\le a<t_b$. These Hasse conditions have exactly the kernel in \eqref{eq:outer-equivalence}. Thus the quotient defining function can be used without first expanding all coefficients of $H_0$. An alternative ground-field recurrence matrix is given in \cref{app:pairing}.

Denote by $B_{\mathrm{aff}}^{\mathrm{eff}}$, $B_{\mathrm{aff}}^{\mathrm{rk}}$ and $B_{\mathrm{aff}}^{\mathrm{str}}$ the minima of the effective, rank-sensitive and rank-free bounds over their respective admissible pairs $(b,J)$. Only pairs with a positive resulting subcode dimension are used. They form finite families because of \eqref{eq:J-admissible}.
\begin{corollary}\label{cor:affine-chain}
These families are nonempty and
\begin{equation}\label{eq:affine-chain}
 d_H(C)\le B_{\mathrm{aff}}^{\mathrm{eff}}(C)
 \le B_{\mathrm{aff}}^{\mathrm{rk}}(C)
 \le B_{\mathrm{aff}}^{\mathrm{str}}(C)
 \le B_\Gamma(f)\le\bsp(C).
\end{equation}
\end{corollary}
\begin{proof}
Take $b\equiv\kappa\pmod m$, $0\le b<m$, and $J=\{0,\ldots,t_b+r_0\}$. Then $r_b=t_b+r_0$, $|J|=r_b+1$, and the largest degree is
$eQ(b+mr_b)=eQ(\kappa+mr_0)=\tau<N$.
The support union is $eE_\Gamma$. This candidate's rank-free bound is exactly $B_\Gamma(f)$. All other comparisons follow from \cref{thm:outer-subcode,thm:envelopes}.
\end{proof}
In the purely origin-multiplicative case $V=0$, $Z=X$, $m=\sigma$ and $\kappa=0$, the test powers are monomials. The structural bound then equals $\tau/s+1$. The rank refinement also uses the relations in the outer quotient.

The bounds in \cref{cor:core-bounds} may be applied independently. Equality in a given rank-free bound requires $\rho_{b,J}=r_b$, $E_{\mathrm{eff}}(W_{b,J})=E_{b,J}$, the MDS property of the subcode after deletion of its zero coordinates, and $d_H(C)=d_H(C\cap W_{b,J})$.

\begin{remark}\label{rem:single-code}
If $f=(X-\alpha)^\tau$ has only one distinct root, then $\alpha\in k^\times$ and $\tau\le p^\eta$. The code contains $(X-\alpha)^j$ for every $\tau\le j<N$, so \cref{prop:single-root} gives
\[
 d_H(C)\le\min_{\tau\le j<N}\prod_r(j_r+1),\qquad j=\sum_rj_rp^r.
\]
If $n=N/p^\eta>1$, the binomial $X^{p^\eta}-\alpha^{p^\eta}$ is a valid nonzero codeword, and $d_H(C)=2$. A weight-one word cannot vanish at a nonzero root. We do not need a general equality assertion for the displayed digit minimum when $n=1$.
\end{remark}
\section{Examples}\label{sec:examples}
\subsection{Irreducible polynomials of order thirteen}
Let $\alpha^3-\alpha=1$ over $\F_3$. The polynomial $X^3-X-1$ has no root in $\F_3$, hence is irreducible. One has $\alpha^9=\alpha-1$ and
\[
 \alpha^{13}=\alpha(\alpha+1)(\alpha-1)=1.
\]
Thus $\alpha$ has order thirteen. Set $\zeta_{13}=\alpha$ and $\xi=-\alpha^7$, of order twenty-six, so $\xi^2=\alpha$. The defining set $\{1,3,9\}$ embeds as $\{2,6,18\}\subseteq\Z/26\Z$, and the three Zech values
\[
 \mathcal Z_\xi(2)=6,\qquad\mathcal Z_\xi(6)=18,\qquad\mathcal Z_\xi(18)=2
\]
recover the translation group $V=\F_3$. The original index-translation stabilizer in $\Z/13\Z$ is trivial. The full geometric group, however, is
\[
 \Gamma=\{(a,(1-a)\alpha+b):a\in\F_3^\times,\ b\in\F_3\}\simeq S_3.
\]
Only its three translations have parameters in $\F_3$. The center coset is the root set itself; $Z=X^3-X-1$, $\kappa=1$, and $H_0=1$. For the origin-marked quotient, $P_0=X^3-X$ and $H_{\mathrm{or}}=Y-1$. The two quotient descriptions thus record different, compatible orbit data.

All four nonzero cyclotomic orbits modulo thirteen are shown in \cref{tab:thirteen}. The polynomial products can be checked by reducing powers of $\alpha$ using $\alpha^3=\alpha+1$.

\begin{table}[H]
\centering
\caption{Irreducible factors of order thirteen over $\F_3$. Here $Q=|V|$ and $m=|A|$.}\label{tab:thirteen}
\begin{tabular}{@{}l l c c c c@{}}
\toprule
Defining set & Monic factor & $Q$ & $m$ & $\wt(f)$ & $B_\Gamma(f)$\\
\midrule
$\{1,3,9\}$ & $X^3-X-1$ &3&2&3&3\\
$\{2,6,5\}$ & $X^3+X^2+X-1$ &1&1&4&4\\
$\{4,12,10\}$ & $X^3+X^2-1$ &1&1&3&4\\
$\{7,8,11\}$ & $X^3-X^2-X-1$ &1&1&4&4\\
\bottomrule
\end{tabular}
\end{table}

For example, at the root $\alpha^2$ of the second polynomial, the Frobenius difference is $1-\alpha$. Its cube is $-\alpha$, which is neither $1-\alpha$ nor $-(1-\alpha)$, so the translation test fails. The nonzero quadratic coefficient also excludes the degree-three form in \cref{thm:degree-p}. The third and fourth polynomials likewise have trivial geometric group. For the third polynomial, $\Gamma=\{1\}$ and $\wt(f)=3<B_\Gamma(f)=4$.

If the multiplicities on $\{1,3,9\}$ are changed from one to three, the ordinary index group is unchanged, but
\[
 f_3=X^9-X^3-1,\qquad \sper(f_3)=3,\qquad
 \Psi_{f_3}(B)=(B^3-B)^3.
\]
The translation points remain $0,1,-1$, now with thickness three in the translation scheme.

\subsection{Irreducible polynomials in characteristic two}
Over $\F_2$, take $f=X^4+X+1$. It has no linear factor and its remainder modulo $X^2+X+1$ is one, so it is irreducible. For a root $\alpha$, put
\[
 b=\alpha^4+\alpha=1,\qquad u=\alpha^2+\alpha.
\]
Then $u^2=u+1$. The plane test gives $V=\{0,1,u,u+1\}=\F_4$, $L_V=X^4+X$, and $L_V(\alpha)=1$. Therefore
\[
 \Gamma=\{(a,(1-a)\alpha+v):a\in\F_4^\times,\ v\in\F_4\}\simeq A_4,
 \qquad\Gamma(\F_2)=\{(1,0),(1,1)\}.
\]
The weight and geometric envelope are both three, whereas the support-period value is five. Its reciprocal $X^4+X^3+1$ has trivial affine group but still has weight three.

The four-element case is not restricted to a linear outer polynomial. Let
\[
 L=X^4+X,\qquad G(Y)=Y^3+Y^2+1,
\]
and set
\[
 f=G(L)=X^{12}+X^9+X^8+X^6+X^3+X^2+1.
\]
The cubic $G$ is irreducible over $\F_2$. Its root $\beta\in\F_8$ has trace one, and $L(x)=\beta$ implies
\[
 x^{64}=x+\beta+\beta^4+\beta^2=x+1.
\]
The square of the $8$-Frobenius permutation on the four-point fiber is a fixed-point-free involution, so that permutation is a four-cycle. Thus $L(X)-\beta$ is irreducible over $\F_8$, and $f$ is irreducible of degree twelve. Here $V=\F_4$, but $m=1$ by \cref{cor:group-shapes}. The basis $1,L,L^2,L^3$ has support envelope of size nine, whereas $f$ has weight seven. The missing outer linear term accounts for cancellations not forced by the full group.

\subsection{A cyclic code of length twenty-six}
Over $\F_3$, set
\[
 L=X^3-X,\qquad f_0=L^2-1=X^6+X^4+X^2-1.
\]
Its roots are $(\alpha+\F_3)\sqcup(-\alpha+\F_3)$, where $\alpha$ is as above. They are distinct roots of $X^{26}-1$, so $C_0=(f_0)\subseteq R_{26,1}$ has dimension twenty. Its geometric group is $\F_3\rtimes\{\pm1\}$: translations are present, their order must divide six, and their scalar normalizer has only the two slopes $\pm1$. Thus
\[
 e=1,\quad Q=3,\quad m=\sigma=2,\quad Z=L,\quad\kappa=0,\quad H_0=Y-1.
\]
The support period is two, and the core generator is
\[
 g(Y)=Y^3+Y^2+Y-1\mid Y^{13}-1.
\]
Although $\wt(f_0)=\bsp(C_0)=4$, choose $b=1$ and $J=\{0,1\}$. Then
\[
 W_{1,J}=\spn\{L,L^3\},\quad E_{1,J}=\{1,3,9\},\quad
 \mathscr H_1=Y-1,\quad\rho_{1,J}=1.
\]
The resulting one-dimensional subcode is generated by
\begin{equation}\label{eq:ternary-word}
 Lf_0=L^3-L=X^9+X^3+X.
\end{equation}
This yields $d_H(C_0)\le3$. A weight-one word cannot vanish at a nonzero root. A weight-two word would, after removal of a monomial and scalar, give $f_0\mid X^j-a$ with $0<j<26$ and $a\in\F_3^\times$. At $\alpha$, its order thirteen forces $a=1$ and $13\mid j$; at $-\alpha$, it also forces $j$ even. Thus $26\mid j$, a contradiction. We obtain
\[
 C_0=[26,20,3],\qquad(C_0)_{\mathrm{sp}}=[13,10,3].
\]
The core version of \eqref{eq:ternary-word} is $(Y-1)g(Y)=Y^4+Y+1$.

For comparison, the Griesmer bound \cite{HP} applied to the ternary $[13,10]$ core also excludes distance at least four, since
\[
 \sum_{i=0}^{9}\ceil{4/3^i}=4+2+8=14>13.
\]

For the same polynomial, $f_0(X+1)=f_0(X)$, but
\[
 \wt(Xf_0)=4,\qquad \wt((X+1)f_0)=8.
\]
Hence the translation preserves the code but not its Hamming weight.

\subsection{A nonzero center}
Over $\F_5$, set
\[
 Z=X-1,\qquad f=Z(Z^2-2)=X^3+2X^2+X+1.
\]
Its roots are $1,1+\beta,1-\beta$, where $\beta^2=2$. The root average is one, so every affine stabilizer fixes this center; in the $Z$-coordinate the slopes are exactly $\pm1$. Hence
\[
 e=Q=\sigma=1,\quad V=0,\quad m=2,\quad\kappa=1,\quad H_0=Y-2.
\]
The two nonrational roots satisfy $(1\pm\beta)^3=2$ and have order twelve. Therefore $f\mid X^{12}-1$ and $C=(f)$ has dimension nine.

For the test spaces, the outer divisors are
\[
 \mathscr H_0(Y)=Y(Y-2),\qquad \mathscr H_1(Y)=Y-2.
\]
The extra factor in $\mathscr H_0$ is required at the distinguished root $1$. For $b=1$, $J=\{0,2\}$, take $R(Y)=Y^2+1$, divisible by $Y-2$ over $\F_5$. Then
\[
 ZR(Z^2)=Z^5+Z=X^5+X-2=f(Z^2+2).
\]
The test powers have support union $\{0,1,5\}$ and detection rank one, so $d_H(C)\le3$. If a binomial $X^j-a$ of degree below twelve vanished on all roots, the root $1$ would give $a=1$ and the order-twelve root would give $12\mid j$. This is impossible. Hence $C=[12,9,3]$, whereas $\bsp(C)=4$.

\subsection{Unequal root multiplicities}
Return to $\F_3$ and $L=X^3-X$, and put
\[
 f_1=(L-1)(L+1)^2=H_1(L),\qquad
 H_1(Y)=Y^3+Y^2-Y-1.
\]
The two root cosets have multiplicities one and two, so $f_1\mid X^{78}-1$, $\deg f_1=9$, and $e=1$. The full group is the translation group $\F_3$; unequal multiplicities exclude the slope $-1$. The support period is one.

Take $J=\{3,9\}$. The relation $Y^9\equiv Y^3\not\equiv0\pmod{H_1}$ gives detection rank one. Consequently
\[
 L^9-L^3=X^{27}+X^9+X^3\in(f_1).
\]
The rank may also be computed from the coefficient sequence of \cref{app:pairing}. Its initial values and recurrence are
\[
 (s_0,s_1,s_2)=(0,0,1),\qquad
 s_{j+3}=-s_{j+2}+s_{j+1}+s_j.
\]
They give
\[
 (s_0,\ldots,s_{11})=(0,0,1,2,2,1,0,0,1,2,2,1),\qquad
 (s_{i+j})_{\substack{0\le i<3\\j\in\{3,9\}}}
 =\begin{pmatrix}2&2\\2&2\\1&1\end{pmatrix}.
\]
This matrix retains the double-root condition. A hypothetical weight-two word forces $26\mid j$ as before, and its root multiplicity must be at least two. A binomial in characteristic three has nonzero-root multiplicity $3^{v_3(j)}$, so also $3\mid j$. Thus $78\mid j$, contradicting degree below seventy-eight. We have
\[
 (f_1)=[78,69,3],\qquad\bsp((f_1))=10.
\]
For comparison, the common powers $f_0^{3^u}$ at length $26\cdot3^u$ have support period $2\cdot3^u$, the same classical core as $f_0$, and parameters $[26\cdot3^u,20\cdot3^u,3]$.

\subsection{The support of a subcode}
For $f=Z^6+Z^3+1\in\F_2[X]$ from \cref{ex:unattained}, one has $f\mid X^{63}-1$. Indeed, the roots satisfy $Z^9=1$, $Z\ne1$, and lie in $\F_{64}$; the corresponding $X$-values are nonzero and the polynomial is squarefree.

For $W=\spn\{1,Z^3,Z^6\}$ the outer polynomial is $Y^2+Y+1$ and a detection matrix is
\[
 M_W=\begin{pmatrix}1&0&1\\0&1&1\end{pmatrix}.
\]
The coefficient row at $X^2$ is $(0,1,1)$, which belongs to $\row(M_W)$. It therefore vanishes on $C\cap W$, although it is present in the test-space envelope. The kernel is generated by $(1,1,1)$, corresponding to $f$, and
\[
 E(W)=\{0,1,2,3,4,6\},\qquad E_{\mathrm{eff}}(W)=\{0,1,3,4,6\}.
\]
Thus $\dim(C\cap W)=1$ and $d_H(C\cap W)=5$, so this subcode gives $d_H(C)\le5$.

\subsection{A family of cyclic codes in odd characteristic}
For an odd prime $p$, define
\[
 L=X^p-X,\quad f_p=L^{p-1}-1,\quad N=p^p-1,\quad
 C_p=(f_p)\subseteq\F_p[X]/(X^N-1).
\]
The roots satisfy $x^p-x\in\F_p^\times$, hence $x^{p^p}=x$, and are nonzero and distinct. Thus $f_p\mid X^N-1$. Also
\[
 f_p=\sum_{j=1}^{p}X^{j(p-1)}-1,\qquad
 \sper(f_p)=p-1,\qquad\bsp(C_p)=p+1,
\]
whereas
\[
 Lf_p=L^p-L=X^{p^2}-2X^p+X
\]
is a weight-three codeword of degree below $N$.

\begin{theorem}\label{thm:prime-family}
Choose $\alpha^p-\alpha=1$ and put
\[
 M_p=\frac{p^p-1}{p-1},\qquad r_p=\ord(\alpha).
\]
Then $r_p\mid M_p$ and
\[
 d_H(C_p)=\begin{cases}2,&r_p<M_p,\\3,&r_p=M_p.\end{cases}
\]
More precisely, for $t>0$ and $a\in\F_p^\times$,
\begin{equation}\label{eq:binomial-family}
 f_p\mid X^t-a\quad\Longleftrightarrow\quad
 (p-1)r_p\mid t\ \text{ and }\ a=1.
\end{equation}
\end{theorem}
\begin{proof}
The roots of $X^p-X-1$ are $\alpha+j$, $j\in\F_p$, and Frobenius cycles through them. Thus $\alpha$ has degree $p$, and
\[
 \alpha^{M_p}=\prod_{j\in\F_p}(\alpha+j)=\alpha^p-\alpha=1.
\]
Also $M_p\equiv1\pmod{p-1}$, so $\gcd(r_p,p-1)=1$. The roots of $f_p$ are exactly
\[
 \{c\alpha+b:c\in\F_p^\times,b\in\F_p\}
 =\{c\alpha^{p^j}:c\in\F_p^\times,0\le j<p\}.
\]
If a binomial vanishes on this set, comparison of $\alpha$ and $c\alpha$ gives $p-1\mid t$. Comparison of $\alpha$ and $\alpha^p$ gives $r_p\mid(p-1)t$, hence $r_p\mid t$ and $a=1$. This proves necessity in \eqref{eq:binomial-family}; the displayed root set proves sufficiency. No weight-one word exists. A weight-two word of degree below $N$ exists exactly when $(p-1)r_p<N$, or $r_p<M_p$. Otherwise the weight-three word above gives equality three.
\end{proof}
The equality $r_p=M_p$ is equivalent to $\alpha$ generating the norm-one subgroup of $\F_{p^p}^\times$. Given the factorization of $M_p$, it is equivalent to
\[
 X^{M_p/\ell}\not\equiv1\pmod{X^p-X-1}
 \qquad\text{for each prime }\ell\mid M_p.
\]

Noncyclic constacyclic examples can be transported explicitly. If $f\mid X^N-1$ over a finite field $F$ and $c\in F^\times$, then
\[
 f_c(X)=c^{\deg f}f(X/c)\mid X^N-c^N.
\]
The substitution multiplies each coefficient coordinate by a nonzero scalar, so it is a monomial isometry. For example, over $\F_9$ take $\theta^2=-1$ and $c=1+\theta$, of order eight. From $f_0$ above one obtains a $[26,20,3]$ constacyclic code with $\lambda=c^{26}=c^2\ne1$. Extension of the coefficient field does not change the distance of a linear code: expansion in a basis of the extension produces nonzero ground-field code components, each supported inside the original word's support, and the reverse inequality follows from inclusion of the ground-field code.

\appendix
\section{Coefficient pairings and recurrence sequences}\label{app:pairing}
We give a recurrence matrix for the outer divisibility ranks in Section~\ref{sec:codes}.

Let $H(Y)=Y^r+\sum_{i=0}^{r-1}h_iY^i\in k[Y]$, with $r\ge1$, and let $A_H=k[Y]/(H)$, writing $y$ for the class of $Y$. Define
\[
 \ell_H(u)=[Y^{r-1}](u\bmod H),
\]
using the unique remainder of degree less than $r$.
\begin{lemma}\label{lem:pairing}
The pairing $(u,v)\mapsto\ell_H(uv)$ on $A_H$ is nondegenerate, and
\begin{equation}\label{eq:trace-pairing}
 \Tr_{A_H/k}(u)=\ell_H(H'(y)u).
\end{equation}
Here the trace is the trace of multiplication by $u$ as a $k$-linear operator. Its associated trace pairing is nondegenerate if and only if $\gcd(H,H')=1$. In that case
\[
 \ell_H(u)=\Tr_{A_H/k}\bigl(u/H'(y)\bigr).
\]
\end{lemma}
\begin{proof}
If a nonzero standard representative $u$ has degree $j$ and leading coefficient $a$, then $uy^{r-1-j}$ has degree $r-1$ without reduction, so $\ell_H(uy^{r-1-j})=a\ne0$.

For the trace identity, divide $H(T)$ by $T-y$ in $A_H[T]$ and write
\[
 H(T)=(T-y)\sum_{i=0}^{r-1}b_iT^i.
\]
Coefficient comparison gives $b_{r-1}=1$, $yb_i=b_{i-1}-h_i$ for $i\ge1$, and $yb_0=-h_0$. The representative of $b_i$ is monic of degree $r-1-i$. Hence $\ell_H(b_i)=\delta_{i0}$, and induction using the recurrence gives
\[
 \ell_H(y^jb_i)=\delta_{ij}\qquad(0\le i,j<r).
\]
The $b_i$ form the dual basis to $y^i$, so the trace of multiplication by $u$ is $\sum_i\ell_H(uy^ib_i)$. Differentiating the displayed factorization in $T$ and setting $T=y$ gives $\sum_i y^ib_i=H'(y)$, proving \eqref{eq:trace-pairing}. Nondegeneracy of $\ell_H$ shows that the radical of the trace pairing is the kernel of multiplication by $H'(y)$. This kernel vanishes precisely when $H'(y)$ is a unit, or $\gcd(H,H')=1$. The final formula follows by substitution.
\end{proof}

\begin{proposition}\label{prop:recurrence}
Put $s_j=\ell_H(y^j)$ for $j\ge0$. Then
\begin{equation}\label{eq:recurrence}
 s_0=\cdots=s_{r-2}=0,\quad s_{r-1}=1,\qquad
 s_{j+r}+\sum_{i=0}^{r-1}h_i s_{j+i}=0\quad(j\ge0).
\end{equation}
The initial zero string is empty for $r=1$. For every finite $J\subseteq\Z_{\ge0}$, the remainder map on $\spn_k\{Y^j:j\in J\}$ has the same kernel and rank as
\begin{equation}\label{eq:hankel-rank}
 (s_{i+j})_{\substack{0\le i<r\\j\in J}}.
\end{equation}
Moreover $H$ is the monic minimum-degree recurrence polynomial of the full sequence.
\end{proposition}
\begin{proof}
The initial values follow from the definition, and the recurrence from $H(y)=0$. By \cref{lem:pairing}, the map
\[
 A_H\longrightarrow k^r,\qquad
 u\longmapsto(\ell_H(u),\ell_H(yu),\ldots,\ell_H(y^{r-1}u))
\]
is an isomorphism. Composing the remainder map with it gives \eqref{eq:hankel-rank}, proving the kernel and rank assertions. If $R$ gives another recurrence, then $\ell_H(y^jR(y))=0$ for every $j\ge0$. The first $r$ equations and nondegeneracy imply $R(y)=0$, so $H\mid R$.
\end{proof}
Applied to $H=\mathscr H_b$, this computes $\rho_{b,J}$ and its parameter kernel. For a general test polynomial $w_j(Y)=\sum_ac_{a,j}Y^a$, the detection column is $(\sum_ac_{a,j}s_{i+a})_{0\le i<r}$. The case $H=1$ has zero target and is handled separately. For a fixed $H$, the same recurrence may be used for several test sets.

For instance, over $\F_2$ let $H=(Y-1)^2=Y^2+1$. The algebra $A_H$ is a dual-number algebra and the ordinary algebra trace is identically zero, whereas $s_j$ is the sequence $0,1,0,1,\ldots$. The matrix
\[
 \begin{pmatrix}s_0&s_1\\s_1&s_2\end{pmatrix}
   =\begin{pmatrix}0&1\\1&0\end{pmatrix}
\]
is invertible. Thus the coefficient pairing, like the Hasse conditions, retains the full multiplicity even where the ordinary trace pairing does not.

\end{document}